\documentclass[a4paper,USenglish]{lipics-v2021}

\usepackage{enumerate}
\usepackage{mathtools}
\usepackage{caption}
\usepackage{float}
\usepackage{tikz}
\usetikzlibrary{arrows,automata,positioning}

\newcommand{\N}{\mathbb{N}}
\newcommand{\E}{\mathbb{E}}
\newcommand{\R}{\mathbb{R}}

\newcommand{\tdist}{\mathit{geom'}(\alpha)}
\newcommand{\cdist}{\mathit{Binom'}}
\newcommand{\fdist}{\mathcal{F}}

\newcommand{\tail}{\mathit{tail}}

\newcommand{\TM}{\textsc{Tail-max}}
\newcommand{\prev}{\mathit{prev}}
\newcommand{\maxT}{\mathit{maxTail(t)}}
\newcommand{\maxV}{\mathit{maxPair}_{\preceq}(t)}
\newcommand{\Tree}{\textsc{Tree}\,}
\newcommand{\Subtree}{\textsc{SubTree}\,}

\usepackage{xcolor}

\usepackage{booktabs}

\usepackage[bibliography=common,appendix=append]{apxproof} %
\newtheoremrep{lemma}{Lemma}
\newtheoremrep{theorem}{Theorem}
\newtheoremrep{corollary}{Corollary}
\newtheoremrep{proposition}{Proposition}

\allowdisplaybreaks[1]

\usepackage{tikz}
\usetikzlibrary{patterns}
\usepackage[precision=2, unit=mm]{lengthconvert}

\title{Optimal RANDAO Manipulation in Ethereum}

\author{Kaya Alpturer}%
       {Princeton University, NJ, USA}%
       {kalpturer@princeton.edu}%
       {https://orcid.org/0000-0003-4843-883X}{}
\author{S. Matthew Weinberg}%
       {Princeton University, NJ, USA}%
       {smweinberg@princeton.edu}%
       {https://orcid.org/0000-0001-7744-795X}{}

\authorrunning{K. Alpturer and S.\,M. Weinberg}

\Copyright{Kaya Alpturer and S. Matthew Weinberg} %

\ccsdesc[500]{Theory of computation~Algorithmic game theory and mechanism design}
\ccsdesc[500]{Information systems~Digital cash}
\ccsdesc[500]{Security and privacy~Distributed systems security}

\keywords{Proof of Stake, Consensus, Blockchain, Ethereum, Randomness manipulation}

\category{} %

\hideLIPIcs

\supplementdetails{Software}{https://github.com/kalpturer/randao-manipulation}

\funding{Supported by NSF CAREER CCF-1942497 and a grant from the Ripple University Blockchain Research Initiative.}
\acknowledgements{We are grateful to Noah Citron for introducing us to the problem, and helpful discussions in early phases of this work. We are also grateful to Yunus Ayd{\i}n, István Seres, Aadityan Ganesh, and anonymous reviewers for feedback on earlier drafts of this work.
}%

\nolinenumbers

\begin{document}

\maketitle

\begin{abstract}
It is well-known that RANDAO manipulation is possible in Ethereum if an adversary controls the
proposers assigned to the last slots in an epoch. We provide a methodology to compute, for any
fraction $\alpha$ of stake owned by an adversary, the maximum fraction $f(\alpha)$ of rounds that a
strategic adversary can propose. We further implement our methodology and compute $f(\cdot)$ for all
$\alpha$. For example, we conclude that an optimal strategic participant with $5\%$ of the stake can
propose a $5.048\%$ fraction of rounds, $10\%$ of the stake can propose a $10.19\%$ fraction of
rounds, and $20\%$ of the stake can propose a $20.68\%$ fraction of rounds.
\end{abstract}

\section{Introduction}
Randomness is an essential component of blockchain protocols. With the invention of Proof of Work
blockchains \cite{nakamoto2008bitcoin}, a major innovation in Bitcoin was to use the
randomness of the
SHA256 function to select the next block proposer. In particular, a participant in the Bitcoin
ecosystem is able to propose a block of their choice with probability proportional to their
computational power. While this system satisfies many desirable properties, it is in many ways not
desirable due to inefficiency. With the move to Proof of Stake blockchain protocols, the dependence
on computation is replaced with stake in the digital currency itself. However, a new source of
randomness is needed to select the next block proposer with probability \emph{proportional to one's
stake}. A major security requirement for this randomness is for it to be verifiable (i.e. everyone
can verify that the block proposer lottery was not rigged) and unpredictable (i.e. before the
lottery happens, no one can know the winner).

Several approaches exist to provide this source of randomness to Proof of Stake blockchain
protocols. One approach is to use an external randomness beacon \cite{fw21}
which achieves similar guarantees as in Bitcoin.
However, implementing such a beacon comes with trust centralization
concerns. A more practical approach is to use protocols that rely on pseudorandom cryptographic
primitives to select block proposers, which are adopted by Proof of Stake blockchains such as
Ethereum.

While the system safety and liveness are not compromised by the randomness mechanism in current
Proof of Stake blockchains, it is well-known that they are susceptible to
manipulation (see \cite{ethresearchblog, vitalikspec}). In particular, incentive-incompatibilities
that result in block-withholding behavior exist in many Proof of Stake blockchain protocols—
analogous to selfish mining for Bitcoin \cite{EyalS14,ZET20,SapirshteinSZ16}.

There is a recent line of work focusing on Proof of Stake incentive incompatibilities
\cite{bnpw19,fw21,FHWY22,BahraniW24,FerreiraGHHWY24},
some of which derive concrete bounds on
optimally manipulating randomness for the Algorand protocol. However, while some
analyses such as \cite{ethresearchblog,eth2book} and simulation-based approaches such as
\cite{alturki2020statistical} conclude that randomness manipulation is likely to be negligible for
Ethereum, it is currently unknown how much more an adversary that \emph{optimally} manipulates
Ethereum can make. In this paper, we focus on answering this question and compute optimal
strategies for randomness manipulation in Ethereum.
Our approach relies on modeling the randomness manipulation game as a
Markov decision process.

\subsection{Brief overview of Proof of Stake Ethereum}
We now briefly cover the relevant details of the Proof of Stake Ethereum protocol.
In the Ethereum protocol, time is divided into epochs, each epoch is divided into 32 slots,
and each slot is 12 seconds. Each epoch is assigned 32 block proposers (one for each slot)
who can construct and broadcast a block to be added to the blockchain at that slot.
If the block proposer fails to do so, the slot is \emph{missed} (no block is added)
and the blockchain moves on to the next slot.

Proof of Stake Ethereum provides randomness by a scheme that maintains a random value called
the RANDAO (also known as  \verb|randao_reveal|) in each block \cite{eth2book}.
As each block is proposed, the previous RANDAO value is mixed using the private key
of the proposer. Since the private key is used to sign the epoch number and is mixed into
the previous RANDAO value by the xor operation, the mixing is \emph{verifiable}.
Moreover, as the signature is assumed to be
uniformly random and the private key is unknown to the public, it is \emph{unpredictable}.
These properties ensure that the only actions available to an adversary in influencing the
RANDAO value is to choose between \emph{broadcasting} or \emph{withholding} a block.

At the end of each epoch, the RANDAO value is used to select a set of 32 new proposers for the next
epoch\footnote{Ethereum actually skips an epoch in this process so the RANDAO value at the end of
epoch $i$ determines the 32 proposers for epoch $i+2$ -- we discuss this in more detail in Section~\ref{sec:randao-game}.}. For example,
if an adversary controls multiple validators and
happens to get assigned to propose in
slot 30 and 31 (the last two slots of an epoch), after slot 29 passes, the
adversary can use the RANDAO value at slot 29 to compute 4 different RANDAO outcomes for the next
epoch. If the adversary withholds both 30 and 31, the RANDAO value remains the same. If the
adversary withholds 30 and broadcasts 31, the RANDAO value is only mixed with the signature of the
proposer of 31, and so on. By precomputing 4 possible outcomes, the adversary is able to
select one of the 4 RANDAO values (at the cost of missing the relevant block rewards).
Similarly, in general, an adversary that controls the last $k$ proposers in an epoch is able to choose
from $2^k$ RANDAO values that determine the proposers for the next epoch.
We call the longest contiguous adversarial slots at the end of an epoch the tail.
With this scheme, it is conceivable that an adversary may strategically withhold their block at
specific slots to win the right to produce \emph{more} blocks in expectation.

\subsection{Main contributions}
Our main technical contributions in this paper are:
\begin{itemize}
    \item We model and formalize the game that an adversary with $\alpha < 1$ proportion of
          the total staked Ethereum plays in manipulating the RANDAO value.
    \item We show that the RANDAO manipulation game can be formulated as a Markov Decision Process, and we show how to significantly reduce the state space so that policy iteration on a laptop quickly converges.\footnote{Our implementation is available here: \url{https://github.com/kalpturer/randao-manipulation}}
    \item We present precise answers to the fraction of slots a strategic player can propose after optimally manipulating Ethereum's RANDAO.
\end{itemize}

\subsection{Related Work}

\noindent\textbf{Manipulating Ethereum's RANDAO.}  The name RANDAO comes from (and the scheme is inspired by) an earlier project~\cite{randao_original}, and its manipulability has been acknowledged and discussed in the Ethereum community~\cite{vitalikblog1, vitalikblog2}. Work of~\cite{alturki2020statistical} focuses on modeling the RANDAO
    mechanism\footnote{The RANDAO model in \cite{alturki2020statistical} is an earlier variation of the RANDAO mechanism that uses a commit-reveal scheme. The induced game, however, is quite similar. }
    in probabilistic rewrite logic
    and evaluating greedy strategies (analogous to $\TM$ of Section~\ref{sub:tailmax}).
    In evaluating the model, they follow a simulation based approach with epoch length 10 and
    report some biasability. In \cite{ethresearchblog}, some probabilistic analysis of how many blocks
    an adversary controlling the last two proposers in the current epoch can get in the next epoch is presented.
    In addition, it is demonstrated that in specific instances, some staking pools had the opportunity to
    control more than half the next epoch. Lastly, \cite{eth2book} shows that the number of proposers an adversary controls at the tail is
    expected to shrink as long as the adversarial stake is less than roughly $1/2$. It also provides a
    strategy analysis that considers tails of length 0 and 1, concluding marginal improvement over the honest.
    These results are consistent with ours, and as expected the reported improvement in rewards
    is less than the optimal strategy we compute. For example, a $25\%$ adversary with their single look-ahead
    strategy makes $2.99\%$ more than the honest while we compute that the optimal strategy
    makes $4.09\%$ more than the honest. In comparison to these works, our work nails down the optimal RANDAO manipulation, and via a principled framework that can accommodate slight modifications (such as epoch length) as well.\\

    \noindent\textbf{Computing Optimal Manipulations.} The most related methodological papers are~\cite{SapirshteinSZ16, FerreiraGHHWY24}, who also compute optimal strategic manipulations.~\cite{SapirshteinSZ16} computes optimal manipulations in Bitcoin's longest-chain protocol, and~\cite{FerreiraGHHWY24} computes optimal manipulations in Algorand's cryptographic self-selection. Our work is methodologically similar, as we also formulate an MDP and use some technical creativity to solve it. On the methodological front, our state-space reduction in Section~\ref{sec:mdp} is perhaps most distinct from prior work.\\

    \noindent\textbf{Manipulating Consensus Protocols, generally.} There is a significant and rapidly-growing body of work on manipulating consensus protocols broadly~\cite{EyalS14, SapirshteinSZ16, KiayiasKKT16, CarlstenKWN16, GorenS19, FiatKKP19, FerreiraW21, FerreiraHWY22, YaishTZ22, YaishSZ23, BahraniW24, FerreiraGHHWY24}. Aside from the aforementioned works, most of these do not compute \emph{optimal} manipulations, but instead understand when profitable manipulations exist. For Ethereum's RANDAO, it is already well-understood that profitable manipulations exist for arbitrarily small stakers, and so the key open problem is how profitable they are (which our work resolves).

\subsection{Roadmap}
In Section~\ref{sec:background}, we briefly cover the relevant background on
Markov Decision Processes. In Section~\ref{sec:randao-game} we formalize the RANDAO
manipulation process as a game. In Section~\ref{sec:mdp}, we formulate the RANDAO
manipulation game as a Markov Decision Process and reduce the state space to a tractable regime. In Section~\ref{sec:eval} and
\ref{sec:solve-opt} we describe how to evaluate and solve for optimal
policies.
Lastly, in
Section~\ref{sec:discussion} we conclude with a discussion on
modeling assumptions, future work, and block miss rates.
We defer some proofs to the appendix.

\section{Preliminaries : Markov Decision Processes} \label{sec:background}
In this section, we quickly review the necessary background for Markov chains
and Markov decision/reward processes.
The material in this section is largely drawn from \cite{howard1960dynamic}
and \cite{puterman2014markov}. These texts can be consulted for a more
detailed treatment.

A \emph{Markov chain} $C = (S,P)$ consists of a set of states $S$ and transition
probabilities $P : S \times S \to \R$. Given a current state $s \in S$,
we transition to the next state with the probability distribution induced by $P(s,\cdot)$.
Rewards can be added to this framework with a reward function $R : S \times S \to \R$
such that if we transition from state $s$ to $s'$, we get $R(s,s')$ reward.
A Markov chain with rewards is a \emph{Markov reward process (MRP)}.

An agent navigating a Markovian system can be modelled using a \emph{Markov
decision process (MDP)}. An MDP is a tuple $M = (S, A, \{P_{a}\}_{a\in A},
\{R_{a}\}_{a\in A})$ where $S$ is a set of states, $A$ is a set of actions,
$P_{a} : S \times S \to \R$ is a transition probability function representing the
probability of individual transitions given an action $a \in A$, and
$R_{a} : S \times S \to \R$ represents the reward of transitioning between
individual states with action $a \in A$.

A \emph{policy} $\pi : S \to A$ is a function specifying
which actions to take given the current state. Once we fix a policy in an MDP,
we get an MRP. We only consider deterministic policies since the standard results
\cite{puterman2014markov} show that in the models we consider, an optimal
deterministic policy exists.

We will be modeling the RANDAO manipulation game as an MDP.
We now introduce some definitions and  properties that will be
useful when we introduce our MDP.

\subsection{Properties of Markovian systems}

One important property concerns whether some states are visited infinitely often.

A state $s$ in a Markov Chain is \emph{recurrent} if, conditioned on currently being at state $s$, the probability of later returning to state $s$ is $1$. If a state is not recurrent, it is called \emph{transient}.
A \emph{recurrent class} of states is a set of recurrent states $\hat{S}$ such that, for all $s \in \hat{S}$, conditioned on being at state $s$:
for all $s' \in \hat{S}$, the probability of visiting $s'$ at a later time is $>0$.

A Markov chain is \emph{ergodic} if it consists of a single recurrent class of states.
Similarly an MDP is ergodic if for every
deterministic policy,\footnote{We only work with stationary policies which are policies that do not change over time. For the processes we consider a stationary optimal policy is guaranteed to exist.}
the Markov chain induced by the policy is ergodic. All MDPs we will consider will be
ergodic so for the rest of this section we assume ergodicity.

A \emph{stationary distribution} $\boldsymbol{\sigma}$ of a transition probability matrix
$\mathbf{P}$ in a Markov chain is defined to be a solution to the following:
\[
  \boldsymbol{\sigma} = \boldsymbol{\sigma}\mathbf{P}
  \qquad \text{ and } \qquad
  \sum_{i}\sigma_{i} = 1
\]

\begin{proposition}[\cite{puterman2014markov}] %
  If a Markov
  chain
  is ergodic, then there exists a unique
  stationary distribution.
\end{proposition}

\subsection{Reward criteria}
Now we define the reward criteria for a Markov decision process $M$.
For our application, \emph{average reward} is more appropriate than
\emph{discounted reward} since we are concerned with the infinite
behavior of the system.

Let $\rho_{\pi,s}(m)$ be a random variable that is equal to the reward at time
$m$ while transitioning the state at time $m$ to the state at time $m+1$
in some MDP when running policy $\pi$ starting at state $s$.

The \emph{average reward} of a policy $\pi$ is defined as:
\[
  \Gamma_{\pi}
  =  \lim_{N \to \infty} \E \left[ \frac{1}{N}
     \sum_{m = 0}^{N} \rho_{\pi}(m) \right]
\]
where we rely on the following result to ignore the initial state.
\begin{proposition}[\cite{puterman2014markov}]
  The average reward for ergodic MDPs is initial state independent.
\end{proposition}

Let $q(s)$ be the expected reward of transitioning from state $s$.
More formally $q(s) = \sum_{s' \in S} R_{\pi(s)}(s,s')P_{\pi(s)}(s,s')$.
Then, $\Gamma_{\pi}$ can be computed using the stationary distribution $\boldsymbol{\sigma}^\pi$
of the Markov chain induced by fixing policy $\pi$.
\[
    \Gamma_{\pi} = \sum_{s \in S} q(s)\sigma^{\pi}_{s}
\]

\noindent\textbf{Value Functions.} In any recurrent process, it is a useful concept to understand the `value' of being in one state over another, due to the potential future rewards. With non-discounted rewards, this requires some subtlety to properly define (because the expected future reward from any state is infinite). One standard method is to define the value of a state as its \emph{average adjusted sum of rewards}:
\[
  v_{\pi}(s)
  = \lim_{N \to \infty}
      \E \left[
        \sum_{m = 0}^{N} (\rho_{\pi,s}(m) - \Gamma_{\pi}).
      \right]
\]

That is, the average adjusted sum of rewards captures the additive difference between an unbounded process starting from state $s$ and iterating $\pi$ and an unbounded process that earns $\Gamma_\pi$ (the average per-round reward of $\pi$) per round.

\begin{lemma}[\cite{howard1960dynamic}]
For an ergodic MDP $M$,
\[
  v_{\pi}(s) + \Gamma_{\pi} = q(s) + \sum_{s' \in S} P_{\pi(s)}(s,s') v_{\pi}(s')
\]
\end{lemma}

Since we can compute $\Gamma_{\pi}$ first given a policy,
this equation determines all $v_{\pi}$ up to an additive constant which we can solve for
after setting $v_{\pi}(s) = 0$ for some $s \in S$.\footnote{The average reward $\Gamma_{\pi}$ is sometimes called \emph{gain}
  and what we call the \emph{value} $v_{\pi}$,
  which is the average adjusted sum of rewards, is sometimes called \emph{bias} in the literature.}

To find the optimal policy with respect to the average reward criterion,
we can run the policy iteration algorithm of \cite{howard1960dynamic, puterman2014markov}.
Starting from an arbitrary policy $\pi_{0}$,
evaluate the policy to compute $\Gamma_{\pi_{0}}$ and $v_{\pi_{0}}$.
Then, a policy improvement step is performed which defines
$\pi_{i+1}(s) := \arg\max_{a \in A} \{q_{a}(s) + \sum_{s' \in S} P_{a}(s,s') v_{\pi_{i}}(s') \}$.
The following Bellman optimality equation guarantees that once this process
stabilizes, we have an average reward optimal policy.
\begin{theorem}[Bellman optimality equation for average-reward MDPs, \cite{howard1960dynamic}]
  If a policy $\pi^{*}$ satisfies the following equations in an MDP,
  then $\pi^*$ is average reward optimal.
  \[
    v_{\pi^{*}}(s) + \Gamma_{\pi^{*}}
    = \max_{a \in A}\left\{
        q_{a}(s) + \sum_{s'\in S} P_{a}(s,s')v_{\pi^{*}}(s')
    \right\}
    \qquad \qquad \forall s \in S
  \]
\end{theorem}

\section{The RANDAO manipulation game} \label{sec:randao-game}
We now review RANDAO in more detail, and formulate the RANDAO manipulation game. RANDAO is a pseudorandom seed that updates every block, and is used to select Ethereum proposers. Below, we use the terminology $R(b)$ to denote the RANDAO value after the $b^{th}$ slot has finished.\\

\noindent\textbf{Updating RANDAO.} The process for updating $R(b)$ is quite simple. If no one proposes a block during slot $b$ of epoch $x$, then $R(b) = R(b-1)$. If a block is proposed during slot $b$,
the proposer must also digitally sign the epoch number $x$
and the hash of this digital signature is XORed with $R(b-1)$ to produce $R(b)$. Note, in particular, that there is a unique private key eligible to propose a block during slot $b$, and therefore the only two possibilities for $R(b)$ are either $R(b-1)$ (if no block is proposed) or $R(b-1)$ XOR
hash(signature of $x$ by proposer for slot $b$).\\

\noindent\textbf{Using RANDAO to seed epochs.} The Ethereum blockchain consists of epochs, where each epoch contains 32 blocks. Within each epoch $t$, a seed $S(t)$ determines which private keys are eligible to propose during each slot. That is, for each of the 32 slots in an epoch, the proposer of that slot is a deterministic function of $S(t)$ (but $S(t)$ is a pseudorandom number). Moreover, if $S(t)$ is a uniformly random number, then each slot proposer is independently and uniformly randomly drawn proportional to stake.

$S(t)$ is set based on RANDAO. Specifically, $S(t)$ is equal to the value of RANDAO at the end of epoch $t-2$. To be extra clear, there have been $32(t-2)$ slots completed by the end of epoch $t-2$, so $S(t):=R(32t-64)$.\\

\noindent\textbf{Rewards.} In practice, proposer rewards involve transaction fees, Maximal Extractable Value (MEV), and any payments made in the Proposer-Builder-Separation (PBS) ecosystem. To streamline analysis, and to be consistent with an overwhelming majority of prior work, we focus on the \emph{fraction of slots} where an adversary proposes.\footnote{Of course, it is an appropriate direction for future work to instead explicitly model transaction fees, MEV, etc. Such modeling would only make an adversary stronger, as their strategy can now depend on the value of each slot~\cite{CarlstenKWN16}.} That is, we consider an adversary who aims to maximize the fraction of slots where they propose.\\

\noindent\textbf{Ideal Cryptography.} It is widely-believed that digital signatures of a previously-unsigned message using an unknown private key (and hashes of previously-unhashed inputs, etc.) are indistinguishable from uniformly-random numbers by computationally-bounded adversaries. However, such cryptographic primitives do not generate truly uniformly-random numbers. For the sake of tractability, and in a manner that is consistent with all prior work studying strategic manipulations of consensus protocols, we consider a mathematical model based on idealized cryptographic primitives (i.e.~that hashing a previously-unhashed input produces a uniformly random number, independent of all prior computed hashes). \\

\noindent\textbf{RANDAO Manipulation Game v1.} We now formally define the RANDAO Manipulation Game (v1). After defining the game, we note its connection to Ethereum's RANDAO, and then proceed to simplify the game. Consistent with an overwhelming majority of prior work, we consider a single strategic manipulator optimizing against honest participants.\footnote{An honest participant proposes a block during every round they are eligible.}

\begin{definition}[RANDAO Manipulation Game v1] The RANDAO Manipulation Game proceeds in epochs $1, 2, \ldots, n, \ldots$. Each epoch has $\ell:=32$ slots. The strategic player has an $\alpha$ fraction of stake.
\begin{itemize}
\item Initialize \textsc{Reward}:=0.
 \item At all times, there is a RANDAO-generated list $R:=\langle R_1,\ldots, R_{32} \rangle \in \{S, H\}^{\ell}$. $R$ denotes the list of $\ell$ proposers based on the current value of RANDAO, and $R_i$ denotes whether the strategic player ($S$) or an honest player ($H$) would propose in an epoch using the current value of RANDAO.
 \item Initialize $R$ so that each coordinate of $R$ is drawn iid, and equal to $S$ with probability $\alpha$.
 \item For each epoch $n:=1,\ldots$
 \begin{itemize}
     \item Store $R^n:=R$ and set the proposers for epoch $n$ equal to $R^n$.
     \item At all times during this epoch, for any set $B$ of slots such that $R^n_i = S$ for all $i \in B$, Strategic Player can compute $R|_B$, which represents how the list of $32$ proposers would update if Strategic Player were to propose a block in exactly slots $B$ and no other blocks are proposed, given that the current RANDAO induces $R$.
     \item For each slot $i:=1,\ldots, \ell$:
     \begin{itemize}
         \item If $R^n_i= H$:
         \begin{itemize}
             \item Update $R$ to redraw each coordinate of $R$ iid, and equal to $S$ with probability $\alpha$.
             \item For all $B$, update $R|_B$ to redraw each coordinate of $R$ iid, and equal to $S$ with probability $\alpha$.
         \end{itemize}
         \item If $R^n_i = S$:
         \begin{itemize}
             \item Strategic Player chooses whether to propose or not.
             \item If they choose to propose: (a) Add $+1$ to $\textsc{Reward}$, and (b) update $R:=R|_i$, and $R|_B:=R|_{B \setminus \{i\}}$ for all $B$.
        \end{itemize}
        \item Store $\textsc{Reward}(n):=\textsc{Reward}$.
     \end{itemize}
 \end{itemize}
 Strategic Player's reward is $\lim\inf_{n \rightarrow \infty} \{\textsc{Reward}(n)/(\ell n)\}$.
\end{itemize}

\end{definition}

Let us now overview the game above, highlight why it captures RANDAO manipulation on Ethereum, and where we've made stylizing assumptions.

\begin{itemize}
    \item First, observe that the epoch's slot proposers are a deterministic function of RANDAO. We have skipped explicitly representing the RANDAO value, and focused only on the resulting proposers in $R$.
    \item Next, observe that every time RANDAO changes \emph{due to an Honest digital signature}, we've randomly redrawn each proposer i.i.d.~and equal to $S$ with probability $\alpha$. This makes two stylizing assumptions.
    \begin{itemize}
        \item First, we've assumed that uniformly RANDAO seed generates proposers iid proportional to stake. This may not be literally true,
        as Ethereum employs a more complicated sampling\footnote{Roughly speaking, for each slot Ethereum shuffles the set of active validators and starts iterating over the shuffled list. A validator is selected to be the proposer for this slot with probability equal to its effective balance over 32.}. However,
        this stylizing assumption has negligible effect for the vast majority of real-world conditions\footnote{As long as most validators have effective balance equal to the maximum, Ethereum essentially selects the proposer using a uniformly random sample. Ethereum's real world conditions match this assumptions
        since the vast majority of validators have maximum effective balance.}.
        \item Second, we've assumed that the hash of an honest digital signature is distributed uniformly at random from the perspective of Strategic Player. This assumes an Ideal hash function and Ideal digital signature (as consistent with prior work~\cite{FerreiraW20, FerreiraHWY22, FerreiraGHHWY24}), although in practice it only holds that the distribution is indistinguishable from uniform to a computationally-bounded adversary.
    \end{itemize}
    \item In our game, the RANDAO value relevant for epoch $t$ is whatever RANDAO is at the end of epoch $t-1$. In Ethereum proper, the RANDAO value relevant for epoch $t$ is at the end of epoch $t-2$. However, we claim our modeling choice is \emph{almost} wlog. Specifically, observe that there are essentially two RANDAO Manipulation Games being played: one on odd epochs, and one on even epochs. That is, the RANDAO value at the end of epoch $2t-1$ determines the proposers for epoch $2t+1$ for all $t$, just as in our RANDAO Manipulation Game. The only distinction to our game is that the RANDAO value at the start of epoch $2t+1$ is \emph{not} equal to the RANDAO value at the end of epoch $2t-1$ (whereas in our RANDAO Manipulation Game, it is) -- the RANDAO value can change during round $2t$. However, \emph{as long as there is at least one Honest proposer during round $2t+1$}, the RANDAO value at the start of epoch $2t+1$ doesn't matter anyway, because it will be reset to uniformly random (at least, from the Strategic Player's perspective).
    \item To elaborate on the previous bullet, \emph{as long as an Honest player proposes in at least one slot in every epoch}, our RANDAO Manipulation Game v1 correctly models all odd Ethereum epochs, and separately correctly models all even Ethereum epochs.
    \item Finally, observe that our Strategic Player receives a reward of one exactly when they propose a block, and their reward is indeed equal to the time-averaged fraction of rounds in which they propose.
    \item To summarize, our stylized game captures RANDAO manipulation in Ethereum with three exceptions: (a) it assumes Ideal cryptography for simplicity of analysis, (b) it assumes proposers in each epoch are drawn
    i.i.d. proportional to fixed stake,
    (c) it assumes that every epoch contains at least one Honest proposer. (a) is a natural assumption consistent with prior works, and essentially abstracts strategic manipulation away from breaking cryptography. The impact of (b) is negligible,
    as the distinction with Ethereum's shuffling and iteration based approach is negligible
    for an essentially uniform\footnote{Here, by uniform, we are referring to the \emph{effective balance} of validators. For Ethereum's current validator set, almost all have maximum effective balance which is 32 ETH.} set of over one million validators.
    (c) is also negligible, as the probability that a Strategic Player could ever induce the next epoch to be the first with no Honest proposers is at most $(2\alpha)^{32}$.
    \footnote{To see this, observe that Strategic Player has at most $2^{32}$ options to seed the subsequent epoch, and for each option the probability that it has no Honest proposers is $\alpha^{32}$. The calculation follows by a union bound. Observe that even for $\alpha = 30\%$, this is $2^{-32}$, meaning we would need to wait $2^{32}$ epochs, or over $150$ years.
    In fact, no strategy can totally control more than a small fraction of the epochs
    when $\alpha$ is less than a reasonable bound such as $\approx \%25$.
    For a discussion on why, see Appendix~\ref{appendix:takeover}.
    }
\end{itemize}

\noindent\textbf{Manipulating RANDAO.} To build intuition, we first give an example of how and why one might manipulate RANDAO. First, imagine that the Strategic Player proposes in slots $25$ and $30$ during epoch $1$. Observe that RANDAO will get reset to uniformly random during epoch $31$, and the RANDAO value between rounds $25$ and $30$ has no impact on any future proposers. Therefore, the Strategic Player gains nothing by skipping these proposal slots.

On the other hand, imagine that the Strategic Player proposes in slots $31$ and $32$. The Strategic Player knows the RANDAO value going into slot $32$, and therefore has two options to set the RANDAO for epoch $2$. If they choose not to propose in slot $32$, they miss out on a one-slot reward, but perhaps this leaves the RANDAO in a favorable place for epoch $2$ as compared to the RANDAO value if they were to propose. In fact, the Strategic Player knows the RANDAO value going into slot $31$, and has four choices between \{propose twice, propose zero times, propose only in $31$, propose only in $32$\}. Each of these will seed a different set of proposers for epoch $2$, and forego a different number of rewards.

This example helps establish that the Strategic Player never benefits from foregoing a proposal before an Honest slot, but has $2^k$ options for the next epoch's RANDAO when they propose the last $k$ slots. We therefore call the largest number of slots $k$ such that the adversary controls the last $k$ slots of an epoch the \emph{tail} of the epoch. The adversary can influence the next epoch only through these slots and
intuitively these slots represent how much predictive power the adversary
holds for the next epoch.\\

\noindent\textbf{Refining the RANDAO Manipulation Game.} Since the length of the tail fully captures the manipulation power of the Strategic Player, we further analyze this and refine our RANDAO Manipulation Game. We first observe that the length of the tail for a single RANDAO draw is distributed according to a roughly geometric distribution.
A tail of length $t$ occurs if we have $t$ slots at the end of an epoch controlled by the strategic player which happens with probability $\alpha^t$, preceded by a single honest slot
which happens with probability $(1-\alpha)$.
We call the remaining non-tail slots that the adversary controls the \emph{count}. The count follows a binomial distribution conditioned on the length of the tail. Specifically:

\begin{itemize}
  \item $\tdist$ is the distribution of the tail given an adversary with stake
        $\alpha$. It is defined such that for $T \sim \tdist$,
        \[ \Pr(T = t) = \begin{cases} (1-\alpha)\alpha^{t} & 0 \leq t < \ell \\
                          \alpha^{\ell} & t = \ell
                        \end{cases}
        \]
  \item $\cdist(\ell - t - 1, \alpha)$ is the distribution of the remaining
        count (how many slots the adversary gets from the non-tail part of the
        epoch) given that the tail is $t$.
        For $C \sim \cdist(\ell - t - 1, \alpha)$,
        \[ \Pr(C = c) = \begin{cases}
            {\ell - t - 1 \choose c} \alpha^c (1-\alpha)^{\ell - t - 1 - c}
              & 0 \leq c < \ell - t \land 0 \leq t < \ell - 1 \\
            1 & c = 0 \land (t = \ell - 1 \lor t = \ell) \\
            0 & c \neq 0 \land (t = \ell - 1 \lor t = \ell)
        \end{cases}
        \]
  \item $\fdist$ is the distribution of $(C,T)$ where
        we first sample $T \sim \tdist$ and then sample
        $C \sim \cdist(\ell - T - 1, \alpha)$.
\end{itemize}

Given our reasoning above, an optimal Strategic Player will always propose during any of the `count' rounds, and will only manipulate the `tail' rounds. In particular, this means that the Strategic Player need not know the full slate of proposers in an epoch, but \emph{only the count and the tail}. With this in mind, we can now refine our RANDAO Manipulation Game v1 to an equivalent RANDAO Manipulation Game v2.

 Using the definitions above, formally, the RANDAO manipulation game $G$ is:
\begin{definition}[RANDAO Manipulation Game v2]\;
  \begin{enumerate}
    \item Initialize $\textsc{Reward}:=0$, and $\textsc{Rounds}:=0$.
    \item Initialize $(c,t)$ drawn from $\fdist$.
    \item For $i = 0$ to $t$,
          \begin{itemize}
            \item For $j = 1$ to ${t \choose i}$,
                  \begin{itemize}
                    \item Sample $(c_{i,j},t_{i,j})$ from $\fdist$.
                  \end{itemize}
          \end{itemize}
    \item The Strategic Player chooses an $(i^*,j^*)$ pair. %
    \item Update $t := t_{i^*,j^*}$, add $c_{i^*,j^*} + t_{i^*,j^*} - i^*$ to \textsc{Reward} and $\ell$ to $\textsc{Rounds}$.
    \item Repeat from step 3.
  \end{enumerate}

  The final payoff is $\lim\inf\{\textsc{Reward}/\textsc{Rounds}\}$.
\end{definition}

RANDAO Manipulation Game v2 is equivalent to RANDAO Manipulation Game v1, after assuming that the Strategic Player optimally proposes during all non-tail slots. Our method of counting the rewards observes that in the next epoch we will always propose during the `count' rounds (and hence just add them directly to our reward), and miss $i$ tail proposals in order to influence the RANDAO (and hence get only $t_{i,j} - i$ slot rewards from the tail).

  Before we proceed with the analysis, we define two sets of interest.
  Let $\mathcal{O}$ be the set of all possible values of $(C_{i,j} - i, T_{i,j})$
  when $(C_{i,j},T_{i,j})$ gets sampled from $\fdist$ for all $0 \leq i \leq t$,
  $1 \leq j \leq {t \choose i}$ for some current tail $t \in \{0,\dots,\ell\}$.
  Given the range of the tail, count and the epoch length,
  $\mathcal{O} = \{(\omega, t) : t \in [0..\ell]
          \land \omega \in [-\ell..\ell]
          \land ((t = \ell \land \omega \leq 0)
            \lor (t < \ell \land \omega \leq \ell - t - 1))
            \}$.
  Let $\Omega$ be the set of all possible multisets of
  observations. More formally, $\Omega$ is the set of
  all multisets $\{(C_{i,j} - i, T_{i,j}) : 0 \leq i \leq \ell, 1 \leq j \leq {t \choose i}\}$
  for some tail length $t$.
  Note that this makes all observation multisets $\mathit{Obs} \in \Omega$
  have size equal to some power of 2.

\section{MDP formulation} \label{sec:mdp}
We can now directly formulate the RANDAO manipulation game as an MDP given the RANDAO Manipulation Game v2.

\begin{definition}[RANDAO MDP $M_G$]
  The Markov decision process representing the RANDAO manipulation game
  is $M_{G} = (S, A, \{P_{\pi}\}_{\pi}, \{R_{\pi}\}_{\pi})$ where
  \begin{itemize}
    \item $S = \{(t,\mathit{Obs}) : t \in \N, 0 \leq t \leq \ell,
          \mathit{Obs} \in \Omega) \}$.
          Each state represents the length of the tail $t$ and the observations
          available to the adversary corresponding to the RANDAO samples.
    \item The action space $A = \mathcal{O}$, each action is selecting a future state from
          the given observations.
    \item $\pi \in \Pi$ is the policy space where $\Pi$ is
          the set of all functions $\pi : \Omega \to \mathcal{O}$
          such that $\pi(\mathit{Obs}) = (\omega,t) \in \mathit{Obs}$.
          $\pi$ chooses on of the sample in $\mathit{Obs}$ to transition towards.
    \item $P_{\pi}$ and $R_{\pi}$ are determined by the process in the game
          where given a state $(t,\mathit{Obs})$,
          we transition to $(t', \mathit{Obs'})$ such that $\pi(\mathit{Obs}) = (\omega',t')$
          and $\mathit{Obs'}$ consists of $2^{t}$ pairs $(c_{i,j} - i, t_{i,j})$
          each sampled using $\fdist$ as in step 2 of the game. The reward of this transition is
          $\omega' + t'$.
  \end{itemize}
\end{definition}

\begin{lemmarep} \label{lem:mg-ergodic}
    $M_{G}$ is ergodic.
\end{lemmarep}
\begin{proof}
  It suffices to observe for any policy $\pi$,
  we can transition from any state to any other state
  in three steps
  with non-zero probability.
  Consider the Markov chain induced by fixing $\pi$ in MDP $M_{G}$.
  Suppose we are at state $(t, \mathit{Obs})$
  and we consider $(t',\mathit{Obs'})$. Let $t^{*} = \log(|\mathit{Obs'}|)$.
  We then observe that the following sequence of transitions have non-zero
  probability:
  \[
    (t, \mathit{Obs})
    \to
    (t_{1},\mathit{Obs}_{1})
    \to
    (t^{*}, \mathit{Obs}_{2})
    \to
    (t', \mathit{Obs'})
  \]
  where
  $\mathit{Obs}_{1}$ is the set of observations where all have tail equal to $t^{*}$
  and
  $\mathit{Obs}_{2}$ is the set of observations where all have tail equal to $t'$.
\end{proof}

It is fairly straight-forward to see that this MDP formulation accurately captures the RANDAO Manipulation Game v2 -- the state captures the point in time after drawing $(c_{i,j}, t_{i,j})$ for all $(i,j)$, and our only action at this point is to choose one such $(i,j)$ and transition, getting reward $\textsc{Reward}$. Note also that every state transition corresponds to exactly an increase of $\ell$ in $\textsc{Rounds}$, so the time-averaged reward in this MDP is exactly the payoff in RANDAO Manipulation Game v2.

Unfortunately, the state space of this MDP is enormous, and we have absolutely no hope of even writing it down (let alone solving it). Luckily since our MDP is ergodic, we can exploit the structure of the optimal policy and drastically simplify the state space.

\paragraph*{Reducing the state space}
We now refine our formulation of the RANDAO manipulation MDP to make it tractable.
We know that for each policy $\pi$, there exists a valuation $v_{\pi} : \{0, \dots, \ell\} \to \R$,
the \emph{average adjusted sum of rewards}. We also know from the Bellman optimality equation
that the optimal policy will take the action maximizing the immediate reward plus the weighted
sum of potential future states with their values. Hence, any plausibly optimal policy, given the
set of samples $\{(c_{i,j},t_{i,j}) : 0 \leq i \leq t, 1 \leq j \leq {t \choose i}\}$, will simply choose the one that maximizes $c_{i,j} - i + t_{i,j} + v_{\pi}(t_{i,j})$. Motivated by this observation we reformulate the RANDAO MDP $M_G$ as the
following reduced state space MDP $M'_{G}$.

Below, intuitively we no longer need to explicitly store all (count, tail) options, because any optimal policy can be fully specified by assigning a value to the tail. So our new state space is simply the tail, but it is now more complex to iterate a transition.

\begin{definition}[RANDAO MDP $M'_{G}$]
  The reduced Markov decision process representing the RANDAO manipulation game
  is $M'_{G} = (S, \Pi, \{P_{\pi}\}_{\pi}, \{R_{\pi}\}_{\pi})$ where
  \begin{itemize}
    \item $S = \{t \in \N : 0 \leq t \leq \ell \}$. Each state represents the
          length of the tail.
    \item $\Pi$ is the policy space where $\Pi$ is the set of all total orders on
          $\mathcal{O}$.
    \item We treat the action space as the same as the policy space. In other words,
          we only consider constant strategies that pick a total order on $\mathcal{O}$.
    \item $P_{\pi}$ and $R_{\pi}$ are determined as follows. Follow RANDAO Manipulation Game v2 in Steps 3-4 (drawing several $(c,t)$s and choosing one), where in Step 4 we choose the future state that is earliest in the total order according to $\pi$. We then transition according to the selected $t$ (this defines $P_{\pi}$) and accumulate reward according to $c+t-i$ (this defines $R_{\pi}$).\footnote{In Section~\ref{sec:eval}, we
          explicitly describe how these transition probabilities and rewards
          are computed.}
  \end{itemize}
\end{definition}

Intuitively, the key difference between $M_G$ and $M'_G$ is at which point in the process we pause and determine a state. In $M_G$, we pause after seeing a large set of (count, tail) pairs and declare this a state. We then make a very simple decision (pick a pair), a very simple reward update (plus count, plus tail, minus number of missed slots), and a fairly simple state transition (draw the new collection of pairs from a known distribution based on the chosen tail).

In $M'_G$, we instead pause after selecting the tail, and declare this a state. We then make a complex decision (pick a total ordering over all plausible pairs), a complex and randomized reward update (sample the set of pairs according to the known distribution based on the state, pick the highest in the total order, and take the reward), and a complex state transition (sample the set of pairs according to the known distribution based on the state, pick the highest in the total order, and take the reward). That is, $M_G$ has a very complicated state space but simple transitions, whereas $M'_G$ has a very complicated action space but simple states. Moreover, we use the Bellman optimality principle to narrow down the plausibly optimal actions for consideration in $M'_G$. We now proceed to establish their equivalence formally.

\begin{lemmarep}
    $M'_{G}$ is ergodic.
\end{lemmarep}
\begin{proof}
  It suffices to observe that under any policy $\pi \in \Pi$ the transition
  probability between any pair of states is positive.
  Suppose that we are running policy $\pi$ and are at state $t \in S$
  and we consider destination tail $t' \in S$.
  Now, if all $2^t$ sampled future states have tail $t'$,
  the next state is $t'$. Since this happens with small but non-zero probability,
  ergodicity holds.
\end{proof}

While we are no longer explicitly keeping track of individual observations in the state,
the optimal policy for $M'_{G}$ still achieves reward equal to the
optimal reward in $M_{G}$ (in expectation).
\begin{proposition}
  If $\pi \in \Pi$ is an optimal policy for $M_{G}$ and
  $\pi' \in \Pi'$ is an optimal policy for $M'_{G}$,
  then $\Gamma_{\pi}(M_{G}) = \Gamma_{\pi'}(M'_{G})$.
\end{proposition}
\begin{proof}
  Suppose $\pi \in \Pi$ is an optimal policy for $M_G$
  and $\pi' \in \Pi'$ is an optimal policy for $M_G'$.
  We first show that given $\pi \in \Pi$ for $M_G$,
  there exists a corresponding policy $\pi^* \in \Pi'$ for $M_G'$
  such that $\Gamma_{\pi}(M_{G}) = \Gamma_{\pi^*}(M'_{G})$.
  Since every optimal policy attains the same expected average reward,
  without loss of generality, assume that $\pi$ satisfies the
  Bellman optimality equation.
  As a consequence, $\pi$ selects the observation $(\omega,t)$ that
  maximizes $\omega + t + v_\pi(t)$.
  This precisely defines a total order on $(\omega,t)$ as $v_\pi$ is fixed.
  Let $\pi^*$ be the total order defined by maximizing $\omega + t + v_\pi(t)$.
  As both $M_G$ and $M_G'$ sample from the same distributions the same number of times,
  $\Gamma_{\pi}(M_{G}) = \Gamma_{\pi^*}(M'_{G})$.

  Hence, we know that given $\pi^*_0$,
  there exists a corresponding policy $\pi^{*'}_0$ playing
  $M_G'$ such that $\Gamma_{\pi^{*}_0}(M_{G}) \leq \Gamma_{\pi^{*'}_0}(M'_{G})$.
  By the optimality of $\pi^*_1$, we also know that
  $\Gamma_{\pi^{*'}_0}(M'_{G}) \leq \Gamma_{\pi^{*}_1}(M'_{G})$.
  Therefore, $\Gamma_{\pi^*_0}(M_{G}) = \Gamma_{\pi^{*}_1}(M'_{G})$
  and the claim holds.
\end{proof}

\section{Evaluating policies} \label{sec:eval}
In this section, we first analyze the policy that only cares about
maximizing the tail length ($\TM$) as an instructive example.
Intuitively, this policy can be implemented by (a) for each
subset of slots that the Strategic Player can choose to propose, computing
the resulting RANDAO value and
hence the next epoch proposer assignments, and (b) the Strategic Player
choosing to propose in the subset of slots that results in the longest tail
in the next epoch.
Subsequently, we describe how to evaluate arbitrary policies in our
Markov decision process $M'_{G}$ formulation of the RANDAO manipulation game.

\subsection{Analyzing the \textsc{Tail-max} policy} \label{sub:tailmax}

The $\TM$ policy can be defined as the policy $\pi$ that given the current
state $t$ and $2^{t}$ samples $(C_{i,j},T_{i,j}) \sim \fdist$, picks $(i,j)$
that maximizes $T_{i,j}$. Note that in case of ties, we pick the transition with
higher reward (i.e. breaking ties in favor of higher $C_{i,j} - i$).

To analyze $\TM$, we are interested in computing the following quantities.
\begin{itemize}
  \item $P_{\tail}(t,t')$ : the probability of transitioning from state $t$ to
        $t'$ when running game $G$ with $\TM$.
  \item $R_{\tail}(t)$ : the expected reward of transitioning from state $t$
        when running game $G$ with $\TM$.
\end{itemize}
Let $\maxT$ be a random variable representing the maximum sampled tail in the current
round of the game. More formally, it is defined as the following:
\[
  \maxT :=
  \max_{\substack{(C_{i,j},T_{i,j}) \sim \fdist \\
   0 \leq i \leq t \\
   1 \leq j \leq {t \choose i}}}
  \{T_{i,j}\}
\]
which is identically distributed as
$\max_{\substack{T_{i} \sim \tdist \\ 0 \leq i \leq 2^{t}-1}}\{T_{i}\}$.
Then, we have the following probabilities.
\begin{align*}
  P_{\tail}(t,t')
  &= \Pr(\maxT = t') \\
  &= \begin{cases}
       \Pr(\maxT \leq t' ) - \Pr(\maxT \leq t' - 1 )  & 0 < t' \leq \ell \\
       \Pr(\maxT \leq t' ) & t' = 0
  \end{cases}
\end{align*}
where we use the following:
\begin{lemmarep}
    $
    \Pr(\maxT \leq T')
    = \begin{cases}
       \left(1-\alpha^{t'+1}\right)^{2^{t}} & 0 \leq t' < \ell \\
       1 & t' = \ell
     \end{cases}
    $
\end{lemmarep}
\begin{proof}
Using the fact that each $\{T_{i}\}$ are i.i.d.,
\begin{align*}
  \Pr(\maxT \leq T')
  &= \Pr_{\substack{T_{i} \sim \tdist \\ 0 \leq i < 2^{t}}} \left( \bigwedge_{0 \leq i < 2^{t}}
       (T_{i} \leq t') \right)\\
  &= \prod_{0 \leq i < 2^{t}} \Pr_{T_{i}\sim\tdist}(T_{i} \leq t') \\
  &= \left( \Pr_{T'' \sim \tdist}(T'' \leq t') \right)^{2^{t}} \\
  &= \begin{cases}
       \left(1-\alpha^{t'+1}\right)^{2^{t}} & 0 \leq t' < \ell \\
       1 & t' = \ell
     \end{cases}
\end{align*}
\end{proof}

The transition reward can be computed in a similar way. Let
$\preceq \in \Pi$ of $M'_{G}$ be defined as $(t,v) \preceq (t',v')$ if and only if $t < t'$ or $(t =
t' \land v \leq v')$. Intuitively, these pairs represent choices that a policy
can make where $t$ is the tail and $v$ is the amount of reward we get from the
rest of the count. The $\TM$ policy picks the maximum such pair according to
$\preceq$. Equality and strict ordering are defined in the usual way. Also let
$\prev_{\preceq}(v,t)$ be defined as the previous pair in the ordering $\preceq$
if it exists and $\bot$ otherwise. Note that $\Pr(\cdot \preceq \bot)$ is
interpreted as 0.

Let $\maxV$ be a random variable representing the maximum tail, and the non-tail reward pair in the current round of the game according to the total order $\preceq$.
More formally, it is defined as the following:
\[
  \maxV :=
  \max^{\preceq}_{\substack{(C_{i,j},T_{i,j}) \sim \fdist \\
                  0 \leq i \leq t \\
                  1 \leq j \leq {t \choose i}}}
  \{(T_{i,j}, C_{i,j} - i)\}
\]
Note that $\preceq$ is a total order and we have the property that $\maxV
\preceq (a,b)$ if and only if $(T_{i,j}, C_{i,j} - i) \preceq (a,b)$ for all $0
\leq i \leq t, 0 \leq j \leq {t \choose i}$. Therefore,
\begin{align*}
  R_{\tail}(t)
  &= \E\left[T' + V' \mid \maxV = (T',V') \right] \\
  &= \sum_{\substack{0 \leq t' \leq \ell \\ -t \leq v' \leq \ell-t'}}
        \Pr(\maxV = (t',v'))(t' + v') \\
  &= \sum_{\substack{0 \leq t' \leq \ell \\ -t \leq v' \leq \ell-t'}}
   (\Pr(\maxV \preceq (t',v))
    - \Pr(\maxV \preceq \prev_{\preceq}(t',v)))(t' + v)
\end{align*}
where we can use the following:
\begin{lemmarep}
Let $\preceq \in \Pi$ of $M'_G$ be the $\TM$ policy.
For $f(x) := \Pr_{T \sim \tdist}(T < x)$,
$g(x) := \Pr_{T \sim \tdist}(T = x)$,
and $h(x,y) := \Pr_{C\sim\text{Binom}(\ell-x-1,\alpha)}(C \leq y)$,
    \begin{align*}
    \Pr(\maxV \preceq (t',v)) &= \prod_{0 \leq i \leq t}
    \left(\Pr_{(C,T)\sim\fdist}((T, C - i)
    \preceq (t',v))\right)^{t \choose i} \\
    \text{and}\qquad\qquad\qquad\qquad\qquad& \\
    \Pr_{(C,T)\sim\fdist}((T, C - i) \preceq (t',v))
    &= \begin{cases}
      f(t') & t' = \ell \land 0 > v' + i \\
      f(t') + g(\ell) & t' = \ell \land 0 \leq v' + i \\
      f(t') + g(t')h(t',v'+i) & \text{otherwise}
    \end{cases}
    \end{align*}
\end{lemmarep}
\begin{proof}
Using standard properties, we observe that
\begin{align*}
  \Pr(\maxV \preceq (t',v))
  &= \Pr_{\substack{(C_{i,j},T_{i,j})\sim\fdist\\
                              0 \leq i \leq t \\
                              1 \leq j \leq {t \choose i}}}
  \left(
        \bigwedge_{\substack{ 0 \leq i \leq t \\
                              1 \leq j \leq {t \choose i}}}
    ((T_{i,j}, C_{i,j} - i) \preceq (t',v)) \right) \\
  &= \prod_{\substack{0 \leq i \leq t \\0 \leq j \leq {t \choose i}}}
    \Pr_{(C,T)\sim\fdist}((T, C - i) \preceq (t',v)) \\
  &= \prod_{0 \leq i \leq t}
    \left(\Pr_{(C,T)\sim\fdist}((T, C - i)
    \preceq (t',v))\right)^{t \choose i}
\end{align*}
and
\begin{align*}
  \Pr&_{(C,T)\sim\fdist}((T, C - i) \preceq (t',v)) \\
  &= \Pr_{(C,T)\sim\fdist}(T < t'
        \lor (T = t' \land C - i \leq v')) \\
  &= \Pr_{T \sim \tdist}(T < t')
    + \Pr_{(C,T)\sim\fdist}(T = t' \land C \leq v'+i) \\
  &= \Pr_{T \sim \tdist}(T < t')
    + \Pr_{T\sim\tdist}(T=t')
      \Pr_{(C,T)\sim\fdist}(C \leq v'+i \mid T = t') \\
  &= \Pr_{T \sim \tdist}(T < t')
    \\ &\quad+ \begin{cases}
      0 & t' = \ell \land 0 > v' + i \\
      \Pr_{T\sim\tdist}(T=\ell) & t' = \ell \land 0 \leq v' + i \\
      \Pr_{T\sim\tdist}(T=t')
      \Pr_{C\sim\text{Binom}(\ell-t'-1,\alpha)}(C \leq v'+i)
      & \text{otherwise}
    \end{cases}
\end{align*}
\end{proof}

We can now compute the stationary distribution
in order to directly compute average reward of the $\TM$ policy. This is a
lower bound to the optimal reward ratio we can obtain from this game.\footnote{Note that $\TM$ does not necessarily outperform Honest -- it could be that in an attempt to increase the tail by one, $\TM$ misses several proposal slots, and yet also does not take good advantage of the increased tail.
However, our results show that $\TM$ does outperform the honest strategy for all $\alpha$.
See Figure~\ref{fig:opt-tm-reward} and \ref{fig:opt-tm-advantage} for a comparison with Honest and the optimal policy. The average tail length of $\TM$, however, serves as an upper bound on the average tail length of any feasible strategy, including the optimum.}

\begin{figure}[H]
  \centering
  \includegraphics[width=\textwidth]{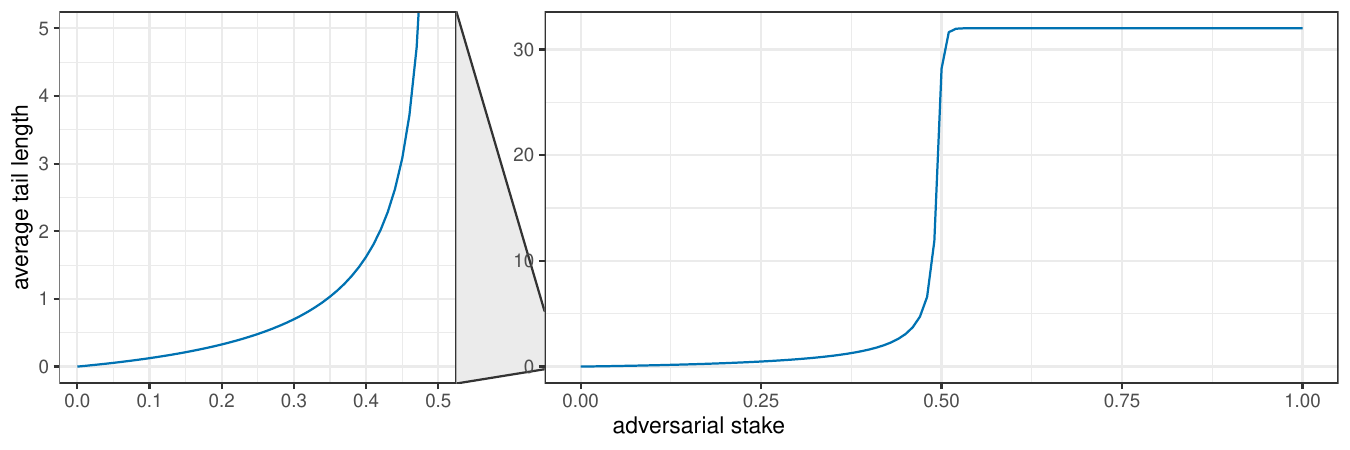}
  \caption{Average tail length attained for each $\alpha$ when running $\TM$. The adversary
  controls a larger tail value as $\alpha$ rises as expected. There is a quick jump
  when approaching $\alpha = 50\%$, indicating that the adversary can
  propose almost all blocks.}
\end{figure}

\subsection{Policy evaluation in the general case}
We proceed similar to the $\TM$ analysis.
Consider policy $\preceq \;\in \Pi$ so it is some total order on $\mathcal{O}$.
Recall that $\maxV$ is a random variable defined (in Subsection~\ref{sub:tailmax})
to be the maximum (tail, non-tail reward) pair given that
the adversary currently controls a tail of length $t$.
Similar to the analysis of $\TM$, we then have
\begin{align*}
  P_{\preceq}(t,t')
  &= \sum_{-\ell \leq v \leq \ell - t'} \Pr(\maxV = (t',v)) \\
  R_{\preceq}(t)
  &= \sum_{\substack{0 \leq t' \leq \ell \\ -t \leq v' \leq \ell-t'}}\Pr(\maxV = (t',v))(t' + v)
\end{align*}
Now, it suffices to describe how to compute the CDF of $\maxV$ since
\[
  \Pr(\maxV = (t',v))
  = \Pr(\maxV \preceq (t',v)) - \Pr(\maxV \preceq \prev_{\preceq}(t',v))
\]

\begin{lemmarep} \label{lem:maxpair-cdf}
Let $\preceq \in \Pi$ be an arbitrary policy in $M'_G$.
    \[\Pr(\maxV \preceq (v,t')) =
    \prod_{0 \leq i \leq t} \left(
    \sum_{\forall (t^{*},v^{*}) \preceq (t',v)} \Pr_{(C,T) \sim \fdist} \left( (T, C) = (t^{*},v^{*} + i) \right)
    \right)^{{t \choose i}}\]
\end{lemmarep}
\begin{proof}
Using standard properties of independent samples, we observe that:
\begin{align*}
  \Pr(\maxV \preceq (v,t'))
  &= \Pr_{\substack{(C_{i,j},T_{i,j}) \sim \fdist \\ 0 \leq i \leq t\\1 \leq j \leq {t \choose i}}}
        \left(\bigwedge_{\substack{0 \leq i \leq t\\1 \leq j \leq {t \choose i}}}
        ((T_{i,j}, C_{i,j} - i) \preceq (t',v) ) \right) \\
  &= \prod_{\substack{0 \leq i \leq t\\1 \leq j \leq {t \choose i}}}
        \Pr_{(C,T) \sim \fdist} \left( (T, C - i) \preceq (t',v) \right) \\
  &= \prod_{\substack{0 \leq i \leq t\\1 \leq j \leq {t \choose i}}}
        \sum_{\forall (t^{*},v^{*}) \preceq (t',v)} \Pr_{(C,T) \sim \fdist} \left( (T, C - i) = (t^{*},v^{*}) \right) \\
  &= \prod_{0 \leq i \leq t} \left(
    \sum_{\forall (t^{*},v^{*}) \preceq (t',v)} \Pr_{(C,T) \sim \fdist} \left( (T, C) = (t^{*},v^{*} + i) \right)
    \right)^{{t \choose i}}
\end{align*}
\end{proof}
Therefore we can also evaluate arbitrary policies by computing the quantities above.
\section{Solving for optimal strategies} \label{sec:solve-opt}
We can now run policy iteration \cite{howard1960dynamic} in our policy
space.
\begin{enumerate}
  \item Start with an arbitrary policy $\preceq \;\in \Pi$.
  \item Compute $P_{\preceq}$ and $R_{\preceq}$.
  \item Compute average reward $\Gamma_\preceq$ using $R_{\preceq}(t)$ and
  the stationary distribution of $P_{\preceq}$.
  \item Determine $v_{t}$ for each $t \in \{0,\dots,\ell\}$
        by setting $v_{0} = 0$ and solving the system of linear equations given
        by $
            \Gamma_\preceq + v_\preceq(t) = R_{\preceq}(t) + \sum_{t' = 0}^{\ell} P_{\preceq}(t,t')v_{t'}
            \qquad \text{ for }
            t \in \{0,\dots,\ell\}
           $.
  \item Let $\preceq'$ be a new total order constructed by sorting each $(\omega,t)$
        pair by the quantities $\omega + t + v_{t}$ in ascending order.
  \item If $\preceq$ is equal to $\preceq'$, we converged and $\preceq$ is optimal.
        Otherwise let $\preceq \;:=\; \preceq'$ and repeat from step 2.
\end{enumerate}
Note that this is equivalent to policy iteration as described in Section~\ref{sec:background} since
sorting by immediate reward plus future state value to get new policy $\preceq$ is equivalent to
picking the maximum at each state.

For $\ell = 32$ and all $\alpha$, policy iteration always converged in less than 10 steps. This enables us to plot the following results on optimal manipulations for RANDAO:
\begin{figure}[H]
  \centering
  \includegraphics[width=\textwidth]{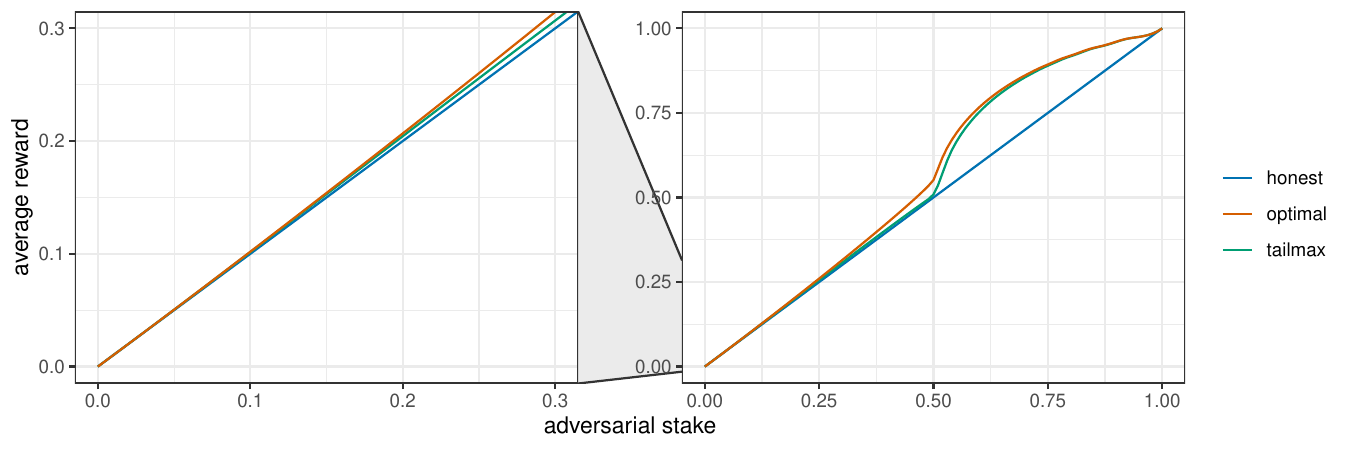}
  \caption{Average reward of the optimal policy and $\TM$ for $\ell = 32$. The figure on the left shows the
  $0 < \alpha \leq 0.3$ range, the figure on the right show the entire range of $0 < \alpha < 1$.}
  \label{fig:opt-tm-reward}
\end{figure}

\begin{figure}[H]
  \centering
  \includegraphics[width=\textwidth]{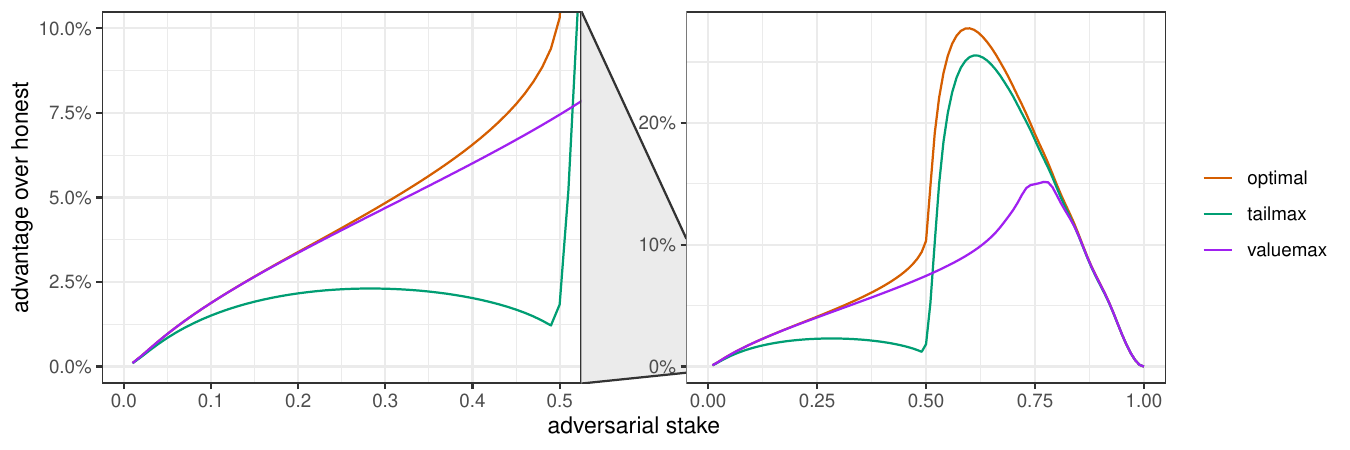}
  \caption{Percentage improvement of the optimal policy and $\TM$ over the honest policy for $\ell = 32$.
  Improvement is defined as $(\text{policy average reward})/(\text{honest average reward}) - 1$.
  We also analyze the strategy \textsc{Value-max} here which we define as the
  strategy that maximizes the reward in the next epoch (chooses the pair that maximizes
  $c_{i,j} + t_{i,j} - i$).}
  \label{fig:opt-tm-advantage}
\end{figure}

\begin{table}[H]
    \centering
    \begin{tabular}{l@{\hspace{3cm}}r}
    \toprule
       $\alpha$ & optimal reward \\ \midrule
        $1\%$  & $1.00107\%$ \\
        $5\%$  & $5.04834\%$ \\
        $10\%$ & $10.18807\%$ \\
        $15\%$ & $15.39960\%$ \\
        $20\%$ & $20.67770\%$ \\
        $25\%$ & $26.02472\%$ \\
        $30\%$ & $31.45164\%$ \\
        $35\%$ & $36.97348\%$ \\
        $40\%$ & $42.62435\%$ \\
        $45\%$ & $48.49184\%$ \\
    \bottomrule \\
    \end{tabular}
    \caption{Average reward of the optimal policy. In expectation, the honest reward is equal to
            $\alpha$. We see in the table that the optimal policy is strictly more profitable.}
    \label{tab:optimal-reward}
\end{table}

A key strength of our methodology is the fact that it is constructive, we
explicitly construct the optimal policy for each $\alpha$.
In order to implement the optimal policy we compute, the values
can be used as in step 5 of the policy iteration routine described above.
For instance, for $\alpha = 0.2$, the optimal policy assigns the
the following values to tails of length 0 to 32 (rounded to two decimal points):
\begin{align*}
(
&0.00, \; 0.90, \; 1.66, \; 2.30, \; 2.86, \; 3.35, \; 3.79, \; 4.19, \; 4.55, \; 4.89, \; 5.21, \; \\ &5.50, \; 5.78, \; 6.05, \; 6.30, \; 6.55, \; 6.78, \; 7.00, \; 7.22, \; 7.43, \; 7.63, \; 7.82, \; \\ &8.01, \; 8.19, \; 8.37, \; 8.54, \; 8.71, \; 8.87, \; 9.03, \; 9.19, \; 9.34, \; 9.49, \; 9.64
)
\end{align*}

\paragraph*{Considering $\ell \neq 32$} \label{sec:numeric}

One benefit of our approach is that it trivially extends to any $\ell$. This allows one to easily answer, for example, whether RANDAO would be more or less manipulable with different epoch lengths and by how much.

Policy iteration runs smoothly for $\ell = 32$ and smaller. However, when we consider $\ell = 64$ or $128$, numerical instability becomes a concern, and our experiments are no longer \emph{provably} accurate. In particular, 64-bit floats we use in our machine introduce precision error that explode when evaluating the expression in Lemma~\ref{lem:maxpair-cdf} with $\ell > 32$.

In order to improve the numerical stability of the expression in
Lemma~\ref{lem:maxpair-cdf}, when $\ell > 32$ we evaluate the inner sum directly to 1 instead of taking the exponential when the sum reaches within $10^{-14}$
of 1 since the error $\epsilon$ introduced by the floating point representation cause $(1 \pm \epsilon)^N$ to become 0 or a large constant for $N \gg 1$. The following figure shows running our evaluation for different $\ell$, using this modification.

\begin{figure}[H]
  \centering
  \includegraphics[width=\textwidth]{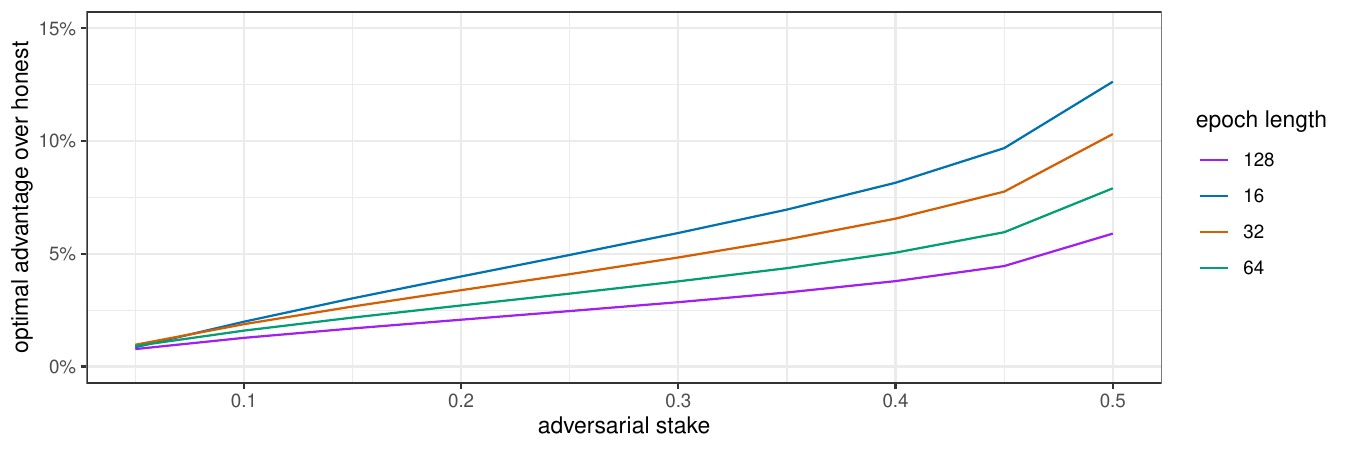}
  \caption{Percentage improvement over honest for $\ell \in \{16,32,64,128\}$.
  Improvement is defined as $(\text{optimal average reward})/(\text{honest average reward}) - 1$.}
\end{figure}

We conjecture that this plot is representative of how the results scale with $\ell$, although unlike our main results the experiments are not provably accurate due to the aforementioned numerical instability. If one desires provable numerical guarantees on these quantities, one would need an analysis of numerical error induced by floating point representations of the machines that run the evaluation.

\section{Discussion} \label{sec:discussion} \nosectionappendix
We model optimal RANDAO manipulation in Proof-of-Stake Ethereum and frame it as an MDP. Our main modeling contribution is getting from RANDAO manipulation in practice to RANDAO Manipulation Game v2, and our key technical insight is getting from there to the reduced RANDAO MDP $M'_G$. From here, simple policy iteration on a laptop suffices to analyze the optimal strategy. Our main results shed light on exactly how manipulable Ethereum's RANDAO is. One could compare our results, for example, to those of~\cite{FerreiraGHHWY24} for Proof-of-Stake protocols based on cryptographic self-selection. For example,~\cite{FerreiraGHHWY24} establishes that a well-connected Strategic Player with $10\%$ of the stake can propose between $10.08\%$ an $10.15\%$ of the rounds in cryptographic self-selection protocols, and our work establishes that a Strategic Player with $10\%$ of the stake can propose a $10.19\%$ fraction of rounds in Ethereum Proof-of-Stake.
While our work introduces methodology to compute these numbers, we leave \emph{interpretation} of their significance to the Ethereum community since many different factors
come into play when designing the consensus mechanism.

A clear direction for future work would be to consider the impact of slot-varying rewards as in~\cite{CarlstenKWN16}. This will clearly increase the manipulability (as now the Strategic Player can use the value of a slot when deciding whether to propose), but it is not obvious by how much. A second direction would be to consider the impact of idiosyncratic details such as Ethereum's sync committees (extra rewards every 256 blocks).

Lastly, we briefly discuss the empirical signature of randomness manipulation.
The results immediately lead to the following question: are there any entities currently
manipulating the RANDAO value?
The signature of such an attack would affect the block miss rates especially around
the tail. Some prior analyses suggest that while there has been
ample opportunities that would result in short term gains for
certain entities, none have been observed to manipulate RANDAO \cite{ethresearchblog}.
In the figure below, the block miss rates by epoch slot index is
displayed from epoch $146876$ to $272341$.
\begin{figure}[H]
  \centering
  \includegraphics[width=\textwidth]{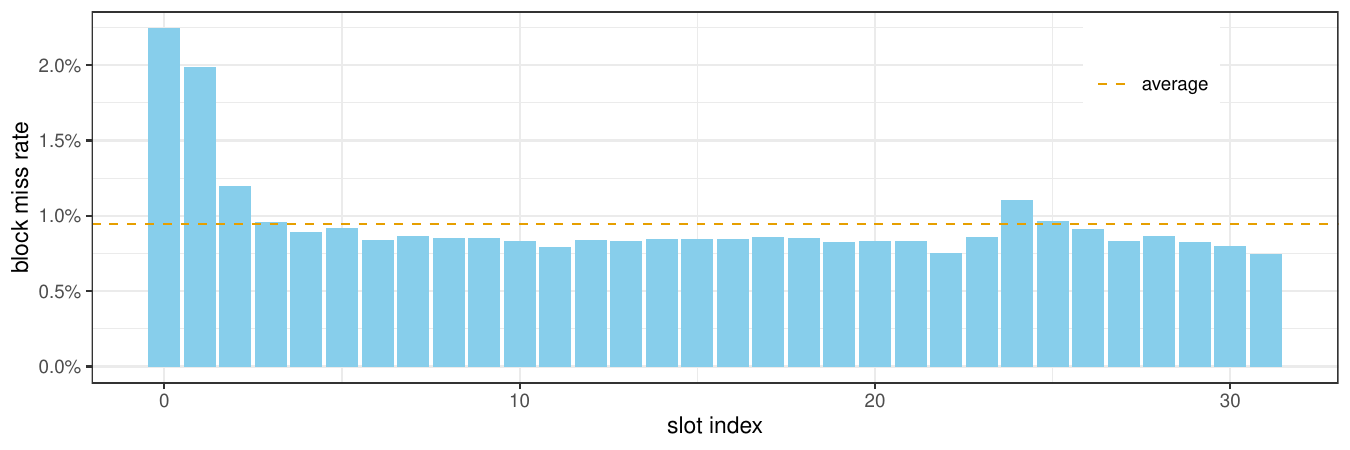}
  \caption{Block miss rate by slot index from epoch $146876$ to $272341$.
  The average block miss rate is $0.9482\%$.}
\end{figure}
Our interpretation of this data is that
slot index 0, 1, and 2 is missed frequently since validators have less time to react to their proposer assignments and
slot index 24 and 25 is missed with slightly higher frequency than the baseline due to votes
crossing the $2/3$ majority.
We do not observe any significant elevation in block miss rates around the tail of the epoch. It is also an interesting direction for future work to examine whether \emph{undetectable} profitable strategies exist for RANDAO manipulation (i.e. strategies that strictly outperform Honest, but produce the same miss rate for all slots).

\bibliography{bib, MasterBib}

\begin{thebibliography}{10}

\bibitem{randao_original}
Randao: A dao working as rng of ethereum, Mar 2019.
\newblock URL: \url{https://github.com/randao/randao/}.

\bibitem{alturki2020statistical}
Musab~A Alturki and Grigore Ro{\c{s}}u.
\newblock Statistical model checking of randao’s resilience to pre-computed reveal strategies.
\newblock In {\em Formal Methods. FM 2019 International Workshops}, pages 337--349, Porto, Portugal, 2020. Springer.

\bibitem{BahraniW24}
Maryam Bahrani and S.~Matthew Weinberg.
\newblock Undetectable selfish mining.
\newblock In {\em {EC} '24: The 25th {ACM} Conference on Economics and Computation}. {ACM}, 2024.

\bibitem{bnpw19}
Jonah Brown-Cohen, Arvind Narayanan, Alexandros Psomas, and S.~Matthew Weinberg.
\newblock Formal barriers to longest-chain proof-of-stake protocols.
\newblock In {\em Proceedings of the 2019 ACM Conference on Economics and Computation}, EC '19, page 459–473, New York, NY, USA, 2019. Association for Computing Machinery.
\newblock \href {https://doi.org/10.1145/3328526.3329567} {\path{doi:10.1145/3328526.3329567}}.

\bibitem{vitalikblog2}
Vitalik Buterin.
\newblock Randao beacon exploitability analysis, round 2, May 2018.
\newblock URL: \url{https://ethresear.ch/t/randao-beacon-exploitability-analysis-round-2/1980}.

\bibitem{vitalikblog1}
Vitalik Buterin.
\newblock Rng exploitability analysis assuming pure randao-based main chain, Apr 2018.
\newblock URL: \url{https://ethresear.ch/t/rng-exploitability-analysis-assuming-pure-randao-based-main-chain/1825/1}.

\bibitem{vitalikspec}
Vitalik Buterin.
\newblock Vitalik’s annotated ethereum 2.0 spec, 2020.
\newblock URL: \url{https://notes.ethereum.org/@vbuterin/SkeyEI3xv#Time-parameters}.

\bibitem{CarlstenKWN16}
Miles Carlsten, Harry~A. Kalodner, S.~Matthew Weinberg, and Arvind Narayanan.
\newblock On the instability of bitcoin without the block reward.
\newblock In {\em Proceedings of the 2016 {ACM} {SIGSAC} Conference on Computer and Communications Security, Vienna, Austria, October 24-28, 2016}, pages 154--167, 2016.
\newblock URL: \url{http://doi.acm.org/10.1145/2976749.2978408}, \href {https://doi.org/10.1145/2976749.2978408} {\path{doi:10.1145/2976749.2978408}}.

\bibitem{eth2book}
Ben Edgington.
\newblock {\em Upgrading Ethereum}.
\newblock Capella edition, 2023.
\newblock URL: \url{https://eth2book.info/}.

\bibitem{EyalS14}
Ittay Eyal and Emin~G{\"u}n Sirer.
\newblock Majority is not enough: Bitcoin mining is vulnerable.
\newblock In {\em Financial Cryptography and Data Security}, pages 436--454. Springer, 2014.

\bibitem{FerreiraHWY22}
Matheus V.~X. Ferreira, Ye~Lin~Sally Hahn, S.~Matthew Weinberg, and Catherine Yu.
\newblock Optimal strategic mining against cryptographic self-selection in proof-of-stake.
\newblock In David~M. Pennock, Ilya Segal, and Sven Seuken, editors, {\em {EC} '22: The 23rd {ACM} Conference on Economics and Computation, Boulder, CO, USA, July 11 - 15, 2022}, pages 89--114. {ACM}, 2022.
\newblock \href {https://doi.org/10.1145/3490486.3538337} {\path{doi:10.1145/3490486.3538337}}.

\bibitem{FerreiraW20}
Matheus V.~X. Ferreira and S.~Matthew Weinberg.
\newblock Credible, truthful, and two-round (optimal) auctions via cryptographic commitments.
\newblock In P{\'{e}}ter Bir{\'{o}}, Jason~D. Hartline, Michael Ostrovsky, and Ariel~D. Procaccia, editors, {\em {EC} '20: The 21st {ACM} Conference on Economics and Computation, Virtual Event, Hungary, July 13-17, 2020}, pages 683--712. {ACM}, 2020.
\newblock \href {https://doi.org/10.1145/3391403.3399495} {\path{doi:10.1145/3391403.3399495}}.

\bibitem{fw21}
Matheus V.~X. Ferreira and S.~Matthew Weinberg.
\newblock Proof-of-stake mining games with perfect randomness.
\newblock In {\em Proceedings of the 22nd ACM Conference on Economics and Computation}, EC '21, page 433–453, New York, NY, USA, 2021. Association for Computing Machinery.
\newblock \href {https://doi.org/10.1145/3465456.3467636} {\path{doi:10.1145/3465456.3467636}}.

\bibitem{FerreiraW21}
Matheus V.~X. Ferreira and S.~Matthew Weinberg.
\newblock Proof-of-stake mining games with perfect randomness.
\newblock In P{\'{e}}ter Bir{\'{o}}, Shuchi Chawla, and Federico Echenique, editors, {\em {EC} '21: The 22nd {ACM} Conference on Economics and Computation, Budapest, Hungary, July 18-23, 2021}, pages 433--453. {ACM}, 2021.
\newblock \href {https://doi.org/10.1145/3465456.3467636} {\path{doi:10.1145/3465456.3467636}}.

\bibitem{FerreiraGHHWY24}
Matheus~V.X. Ferreira, Aadityan Ganesh, Jack Hourigan, Hannah Hu, S.~Matthew Weinberg, and Catherine Yu.
\newblock Computing optimal manipulations in cryptographic self-selection proof-of-stake protocols.
\newblock In {\em {EC} '24: The 25th {ACM} Conference on Economics and Computation}. {ACM}, 2024.

\bibitem{FHWY22}
Matheus~V.X. Ferreira, Ye~Lin~Sally Hahn, S.~Matthew Weinberg, and Catherine Yu.
\newblock Optimal strategic mining against cryptographic self-selection in proof-of-stake.
\newblock In {\em Proceedings of the 23rd ACM Conference on Economics and Computation}, EC '22, page 89–114, New York, NY, USA, 2022. Association for Computing Machinery.
\newblock \href {https://doi.org/10.1145/3490486.3538337} {\path{doi:10.1145/3490486.3538337}}.

\bibitem{FiatKKP19}
Amos Fiat, Anna Karlin, Elias Koutsoupias, and Christos~H. Papadimitriou.
\newblock Energy equilibria in proof-of-work mining.
\newblock In {\em Proceedings of the 2019 {ACM} Conference on Economics and Computation, {EC} 2019, Phoenix, AZ, USA, June 24-28, 2019.}, pages 489--502, 2019.
\newblock \href {https://doi.org/10.1145/3328526.3329630} {\path{doi:10.1145/3328526.3329630}}.

\bibitem{GorenS19}
Guy Goren and Alexander Spiegelman.
\newblock Mind the mining.
\newblock In {\em Proceedings of the 2019 {ACM} Conference on Economics and Computation, {EC} 2019, Phoenix, AZ, USA, June 24-28, 2019.}, pages 475--487, 2019.
\newblock \href {https://doi.org/10.1145/3328526.3329566} {\path{doi:10.1145/3328526.3329566}}.

\bibitem{howard1960dynamic}
R.A. Howard.
\newblock {\em Dynamic programming and Markov processes}.
\newblock Technology Press of Massachusetts Institute of Technology, 1960.

\bibitem{KiayiasKKT16}
Aggelos Kiayias, Elias Koutsoupias, Maria Kyropoulou, and Yiannis Tselekounis.
\newblock Blockchain mining games.
\newblock In {\em Proceedings of the 2016 {ACM} Conference on Economics and Computation, {EC} '16, Maastricht, The Netherlands, July 24-28, 2016}, pages 365--382, 2016.
\newblock URL: \url{http://doi.acm.org/10.1145/2940716.2940773}, \href {https://doi.org/10.1145/2940716.2940773} {\path{doi:10.1145/2940716.2940773}}.

\bibitem{nakamoto2008bitcoin}
S.~Nakamoto.
\newblock Bitcoin: A peer-to-peer electronic cash system.
\newblock \url{http://www.bitcoin.org/bitcoin.pdf}, 2008.

\bibitem{puterman2014markov}
M.L. Puterman.
\newblock {\em Markov Decision Processes: Discrete Stochastic Dynamic Programming}.
\newblock Wiley Series in Probability and Statistics. Wiley, 2014.

\bibitem{SapirshteinSZ16}
Ayelet Sapirshtein, Yonatan Sompolinsky, and Aviv Zohar.
\newblock Optimal selfish mining strategies in bitcoin.
\newblock In {\em Financial Cryptography and Data Security - 20th International Conference, {FC} 2016, Christ Church, Barbados, February 22-26, 2016, Revised Selected Papers}, pages 515--532, 2016.
\newblock \href {https://doi.org/10.1007/978-3-662-54970-4\_30} {\path{doi:10.1007/978-3-662-54970-4\_30}}.

\bibitem{ethresearchblog}
Toni Wahrstätter.
\newblock Selfish mixing and randao manipulation, Jul 2023.
\newblock URL: \url{https://ethresear.ch/t/selfish-mixing-and-randao-manipulation/16081}.

\bibitem{YaishSZ23}
Aviv Yaish, Gilad Stern, and Aviv Zohar.
\newblock Uncle maker: (time)stamping out the competition in ethereum.
\newblock In Weizhi Meng, Christian~Damsgaard Jensen, Cas Cremers, and Engin Kirda, editors, {\em Proceedings of the 2023 {ACM} {SIGSAC} Conference on Computer and Communications Security, {CCS} 2023, Copenhagen, Denmark, November 26-30, 2023}, pages 135--149. {ACM}, 2023.
\newblock \href {https://doi.org/10.1145/3576915.3616674} {\path{doi:10.1145/3576915.3616674}}.

\bibitem{YaishTZ22}
Aviv Yaish, Saar Tochner, and Aviv Zohar.
\newblock Blockchain stretching {\&} squeezing: Manipulating time for your best interest.
\newblock In David~M. Pennock, Ilya Segal, and Sven Seuken, editors, {\em {EC} '22: The 23rd {ACM} Conference on Economics and Computation, Boulder, CO, USA, July 11 - 15, 2022}, pages 65--88. {ACM}, 2022.
\newblock \href {https://doi.org/10.1145/3490486.3538250} {\path{doi:10.1145/3490486.3538250}}.

\bibitem{ZET20}
Roi~Bar Zur, Ittay Eyal, and Aviv Tamar.
\newblock Efficient mdp analysis for selfish-mining in blockchains.
\newblock In {\em Proceedings of the 2nd ACM Conference on Advances in Financial Technologies}, AFT '20, page 113–131, New York, NY, USA, 2020. Association for Computing Machinery.
\newblock \href {https://doi.org/10.1145/3419614.3423264} {\path{doi:10.1145/3419614.3423264}}.

\end{thebibliography}

\begin{toappendix}
\section{Bounds on RANDAO takeover} \label{appendix:takeover}
In this appendix, we generalize the RANDAO manipulation game to model the RANDAO mechanism
in Ethereum more closely. However, as we later argue, the simpler model already
approximates RANDAO manipulation, and the generalized game essentially
converges to the initial game for our purposes. This motivates our definition in the main body
of the paper further, and shows that modeling additional details of the RANDAO
process that we consider in this section is not necessary.

In Ethereum, we actually have two interleaved instances of the RANDAO manipulation game
where the RANDAO value at the end of epoch $i$ determines proposer assignments for
epoch $i+2$, the RANDAO at the end of epoch $i+1$ determines epoch $i+3$, and so on.
Moreover, to see what could happen in Ethereum but not in the RANDAO manipulation game $G$,
consider the following scenario:
Suppose we reach the end of epoch $i-1$. Now, we know the proposer assignment for epoch's
$i$ and $i+1$. Further assume that as the adversary, we have control over a tail of length
$k$ in epoch $i$, and we control the entirety of epoch $i+1$  (so a tail of length $\ell$.)
Adapting the RANDAO game we presented so far, we would get $2^k$ choices of different
assignments for epoch $i+2$ and \emph{independently} $2^\ell$ choices for assignments
in epoch $i+3$. However, in Ethereum, since we control the entirety of epoch $i+1$,
for \emph{each} $2^k$ choices of assignments for epoch $i+2$, we actually get a new
set of $2^\ell$ samples for epoch $i+3$. Since there is no proposer we don't control in
epoch $i+1$, we can precompute the options we will get at the end of epoch $i+1$
before we decide in epoch $i$.

In fact, as the adversary we can see arbitrarily into the future value of RANDAO
provided that we get lucky and the samples we get retain our full control of successive epochs.
We now model all possible future paths as a tree with each node representing
the current size of the adversarial tails as a pair,
the tails of epoch $i$ and $i+1$.
Each edge is labeled with the reward of transitioning to the new epoch assignments.
Note that in this tree, an edge from $(t,t')$ to $(t'',t''')$ implies that
$t' = t''$ since epoch assignments are interleaved.

We now give the definition of $\Tree(t,t')$:
\begin{definition}[$\Tree(\cdot,\cdot)$]\;
\begin{enumerate}
    \item The root node is $(t,t')$.
    \item If $t' = \ell$, then
        \begin{itemize}
        \item For $i$ from $0$ to $t$, and $j$ from $1$ to ${t \choose i}$,
            sample $(c_{i,j},t_{i,j})$ from $\fdist$.
        \item For each $(i,j)$,
	      compute and add $\Subtree(t',t_{i,j})$ by
	      connecting $(t,t')$ to subtree root $(t',t_{i,j})$
	      labeled with reward $c_{i,j} + t_{i,j} - i$.
        \end{itemize}
    \item If $t' < \ell$, then
    \begin{itemize}
        \item For $i$ from $0$ to $t$, and $j$ from $1$ to ${t \choose i}$,
          sample $(c_{i,j},t_{i,j})$ from $\fdist$.
        \item For each $(i,j)$, add node $(t', t_{i,j})$ and
              connect $(t,t')$ to $(t', t_{i,j})$.
              Label this edge with reward $c_{i,j} + t_{i,j} - i$.
    \end{itemize}
\end{enumerate}
\end{definition}

\noindent where we define $\Subtree(t,t')$ recursively as:
\begin{definition}[$\Subtree(\cdot,\cdot)$]\;
\begin{enumerate}
    \item The root node is $(t,t')$.
    \item If $t < \ell$ and $t' \leq \ell$,
          then $\Subtree(t,t')$ simply has a single node, its root.
    \item If $t = \ell$ and $t' \leq \ell$, then
        sample from $\fdist$ exactly as in step 2 of
                  the definition of $\Tree(t,t')$.
\end{enumerate}
\end{definition}

Then the new game can be defined as:
\begin{definition}[the generalized RANDAO manipulation game $G_r$]\;
  \begin{enumerate}
    \item We start at a pair of tails $(t,t') := (t_{0},t'_0)$.
    \item Recursively generate $\Tree(t,t')$ using the definition above.
    \item Choose a node $(t^*,{t^*}')$ and consider the
          unique $k$-step path from the root of
          $\Tree(t,t')$ to $(t^*,{t^*}')$.
    \item Update the current state $(t,t')$ appropriately in $k$
          rounds, eventually reaching $(t,t') := (t^*,{t^*}')$
          and getting a total reward equal to the sum over the path.
          Repeat from step 2.
  \end{enumerate}
\end{definition}

We observe that the general game is equivalent to two interleaved instances of the initial
game if we restrict to always choosing a node from the first level of the $\Tree(\cdot,\cdot)$.
Therefore we can directly run two interleaved instances of any policy $\pi$
playing the game $G$ in $G_r$. Since the reward structure is also the same at the first level
of the $\Tree$, we get the following corollary immediately as a consequence of the definition.
\begin{corollary}
    A policy $\pi$ for the RANDAO manipulation game $G$
    can be directly adapted into a policy $\pi'$ for the general RANDAO manipulation game
    $G_r$ such that they earn the same average reward.
\end{corollary}

\subsection{Bounding the height of the tree}

Right now, the tree might have leaf nodes of the form $(t,\ell)$ that immediately start a
new tree once the adversary reaches the leaf node. For reasons that would be clear in the
next subsection, we now define a modified tree definition $\Tree'$ such that
whenever we reach a terminal path $\cdots \to (32,t) \to (t,32)$ in the expansion of a subtree
where $t < \ell$ and therefore
$(t,32)$ becomes a leaf node, we replace $t$'s with $\ell$'s and keep expanding the
tree. Clearly, the expected height of $\Tree'$ is greater than $\Tree$.
From here on, we use this definition.

Let $H(T)$ be defined as a random variable denoting the height of the (sub)tree $T$
which is the maximum length over root to leaf paths.

\begin{lemma} \label{lem:tree-prob}
$\Pr(H(\Tree(t,\ell)) \geq \lambda) \leq ((2\alpha)^{\ell}
+ (1-\alpha^\ell) (4\alpha)^\ell)^{\lambda-2}$
for $\lambda \geq 2$.
\end{lemma}
\begin{proof}
  We first observe that it suffices to consider $H(\Tree(\ell,\ell))$.
  Since $\Tree(t,\ell)$ is identically distributed with
  a subtree of $\Tree(\ell,\ell)$ and each $\Subtree(\cdot,\cdot)$
  that get sampled in $\Tree(t,\ell)$ has height $\geq 0$,
  $\Pr(H(\Tree(t,\ell)) \geq \lambda) \leq \Pr(H(\Tree(\ell,\ell)) \geq \lambda)$.

  Let $H := H(\Tree(\ell,\ell))$ for notational convenience.
  We now argue by induction that $\Pr(H \geq \lambda)
  = ((2\alpha)^{\ell} + (1-\alpha^\ell) (4\alpha)^\ell)^{\lambda-2}$.
  In the base case, we have $\Pr(H \geq 2) = 1$ since by definition, $\Tree(\ell,\ell')$
  has height at least 2.

  Next we observe that for a (sub)tree to expand one more level, we need
  successive epochs to have tail length equal to $\ell$ in which case by
  step 3 of the definition of $\Subtree(\cdot,\cdot)$, we append one more
  level to the tree to reach a height of at least one more, or
  the next tail length is $< \ell$ but the one following it is $\ell$ in which case we
  still expand with the new definition.
  Suppose we are at $(\ell,\ell)$ which are epochs $i$ and $i+1$.
  For epoch $i+2$, we get $2^\ell$ trials each with success probability
  $\leq \alpha^{\ell} + (1-\alpha^\ell) (2\alpha)^\ell$ using a union bound.
  Let $S$ denote the number of successes.
  We also define $H_j$ as the height of the subtree from each successful branch.
  Using the fact that $\Pr(H_j \geq \lambda) \leq \Pr(H \geq \lambda)$ as before,
  we obtain the following:
    \begin{align*}
        \Pr(H \geq \lambda)
	    &=    \sum_{k=1}^{2^{\ell}} \Pr(S = k) \Pr(\bigcup_{j=1}^{k} (H_j \geq \lambda - 1)) \\
        &\leq \sum_{k=1}^{2^{\ell}} \Pr(S = k) \sum_{i=1}^k \Pr(H \geq \lambda-1))
        && \text{union bound} \\
        &= \sum_{k=1}^{2^{\ell}} k\Pr(S = k)\Pr(H \geq \lambda-1)) \\
        &= \Pr(H \geq \lambda-1)) \sum_{k=1}^{2^{\ell}} k\Pr(S = k) \\
        &\leq \Pr(H \geq \lambda-1)) ((2\alpha)^{\ell} + (1-\alpha^\ell) (4\alpha)^\ell) \\
	      &= ((2\alpha)^{\ell} + (1-\alpha^\ell) (4\alpha)^\ell)^{\lambda-2} && \text{inductive hypothesis}
    \end{align*}
Therefore the claim holds.
\end{proof}

\begin{proposition} \label{prop:exp-tree}
  For $(2\alpha)^{\ell} + (1-\alpha^\ell) (4\alpha)^\ell < 1$,
  $$\E[H(\Tree(t,\ell))] \leq 1 +
  \frac{1}{1 - ((2\alpha)^{\ell} + (1-\alpha^\ell) (4\alpha)^\ell)}$$
\end{proposition}
\begin{proof}
  Let $H := H(\Tree(t,\ell))$.
  \begin{align*}
    \E[H]
    &= \sum_{k = 2}^\infty k \Pr(H = k) \\
    &= \Pr(H \geq 2) + \sum_{k = 2}^\infty \Pr(H \geq k) \\
    &\leq 1 + \sum_{k=2}^\infty ((2\alpha)^{\ell} + (1-\alpha^\ell) (4\alpha)^\ell)^{k - 2}
    && \text{Lemma~\ref{lem:tree-prob}} \\
    &= 1 + \frac{1}{1 - ((2\alpha)^{\ell} + (1-\alpha^\ell) (4\alpha)^\ell)}
    && \text{since $(2\alpha)^{\ell} + (1-\alpha^\ell) (4\alpha)^\ell < 1$}
  \end{align*}
\end{proof}

We note here that $(2\alpha)^{\ell} + (1-\alpha^\ell) (4\alpha)^\ell < 1$ holds for $0 \leq \alpha \leq 0.24$
which is the range we consider for our calculations.

\subsection{Bounding the reset time}
Consider an instance of the game $G_r$ where we start from some state $(t, t')$
while running some policy. We call the number of rounds it takes
to reach a state with $(t,0)$ the reset time. Let $\tau(t,t')$ be a random variable
representing the reset time from state $(t,t')$.

We now would like to bound $\E[\tau(t,\ell)]$.

\begin{lemma} \label{lem:reset-bound}
  $\E[\tau(t,\ell)] \leq X(t,\ell)$ for $(2\alpha)^{\ell} + (1-\alpha^\ell) (4\alpha)^\ell < 1$
  where $X(t,\ell)$ is the maximum solution to the following system of linear equations:
  for $t \in \{0,\dots,\ell\}$ and $t' \in \{1,\dots,\ell-1\}$,
  \begin{align*}
    x_{(t,\ell)} &= 1 + \frac{1}{1 - ((2\alpha)^{\ell} + (1-\alpha^\ell) (4\alpha)^\ell)} + x{(\ell-1,\ell-1)} \\
    x_{(t,t')} &= 1 + \sum_{t'' = 0}^{\ell} \Pr(t \to t'') x_{(t',t'')} \\
    x_{(t,0)} &= 0
  \end{align*}
  and $\Pr(t \to t'')$ is the $\TM$ transition probabilities from Section~\ref{sub:tailmax}.
\end{lemma}
\begin{proof}
  We start by considering a policy that would maximize the reset time
  (while ignoring reward) to get an upper bound on the reset time of any policy in
  expectation.

  Suppose we are at some state $(t',t'')$ where $t'' < \ell$.
  Then, for the next $2^{t'}$ options, clearly choosing the longest tail is best since
  a longer tailed state samples identically distributed future states plus additional samples.
  If $t'' = \ell$, then we would sample $\Tree'(t',\ell)$ and choose a path. Note that it is now not
  clear what path would lead to the longest reset time since we might prefer a shorter path
  that ends at a longer tail state compared to the longest path in the tree.
  To solve this issue, we modify the game such that when $t'' = \ell$ and we sample the tree,
  at the end, each leaf node is replaced with $(\ell-1,\ell-1)$. We also note that
  we are using the modified tree definition $\Tree'$.
  This modification only gives the adversary more options while keeping
  the previous options the same, hence any instance that could happen in the previous
  definition is still possible, and this modification only increases the reset time.
  With the new definition, now it is clear that choosing the maximum length path in the tree
  gives us the greatest reset time.

  Now that the policy maximizing the reset time is fixed, we can analyze it to get the upper
  bound. Suppose $\tau'$ refers to the reset time of this new process and the transition
  probabilities used below are for the tail maximizing policy we just discussed above.
  Consider the state $(t,t')$. Clearly, the expected reset time when $t' = 0$ is 0.
  If $0 < t' < \ell$, then we know that the expected reset time is equal to one more than the expected
  reset time of the next state. Hence,
  \[
    \E[\tau'(t,t')] = 1 + \sum_{t'' = 0}^{\ell} \Pr((t,t') \to (t',t'')) \E[\tau'(t',t'')]
    \qquad \text{and} \qquad
    \E[\tau'(t,0)] = 0
  \]
  Note that the transition probability here when $t' < \ell$ simply becomes the transition probability
  of choosing the maximum tailed state out of $2^t$ samples. This is exactly the
  transition probabilities from Section~\ref{sub:tailmax}.
  Lastly, if $t' = \ell$, we know that we start expanding the tree which due to our modification ends
  at $(\ell - 1, \ell -1)$. By Proposition~\ref{prop:exp-tree}, $\E[H(\Tree'(t,\ell))] \leq
  1 + 1/(1 - ((2\alpha)^{\ell} + (1-\alpha^\ell) (4\alpha)^\ell))$.
  Therefore we get the following:
  \[
    \E[\tau'(t,\ell)] = 1 + 1/(1 - ((2\alpha)^{\ell} + (1-\alpha^\ell) (4\alpha)^\ell)) + \E[\tau'(\ell-1,\ell-1)]
  \]
  We have now written a system of linear equations that represent the expected reset time from each state.
  Since we overestimated the reset time from state $(t,\ell)$, we have only made the reset time larger
  and the solution to this system of linear equation will be an upper bound on the expected reset time.

  Let $X(t,\ell)$ be the maximum of the set of solutions to the system of linear equations above
  which can be computed efficiently.
\end{proof}

Solving the system of linear equations above, we obtain the results in Figure~\ref{fig:reset-time}
for $X(t,32)$.
\begin{figure}[htp]
    \centering
    \includegraphics[width=0.8\textwidth]{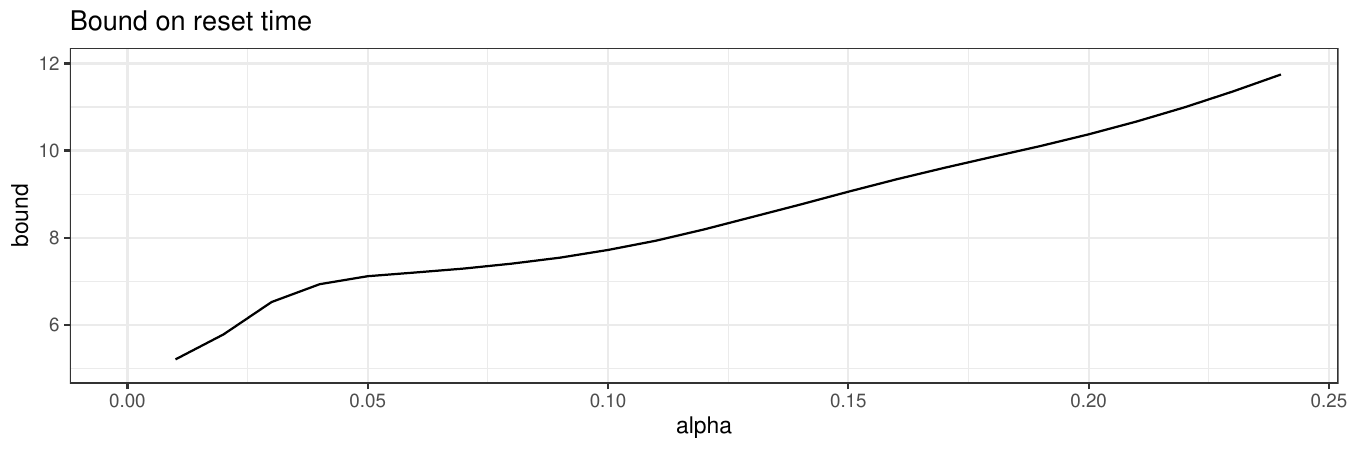}
    \caption{Plot of the reset time bound $X(t,32)$}
    \label{fig:reset-time}
\end{figure}

\subsection{Bounding the error term}
Lastly, we need to show that for any strategy playing the general RANDAO manipulation game,
there exists a strategy in our main model, the interleaved RANDAO manipulation game (where
two instances of our main model are interleaved) that approximately earns
the same average reward in expectation. This would show that considering optimal strategies
in the RANDAO manipulation game is sufficient.

We will proceed by a reformulation of the general RANDAO manipulation game into two explicit phases that
make the distinction with the simpler model explicit and allow us to bound the difference in the average
reward earned by both. We keep track of rewards and rounds from two different phases of the game separately.
Consider the following version of the interleaved RANDAO manipulation game:
\begin{definition}[subgame formulation of the general RANDAO manipulation game]
From some initial state $(t,t')$ and variables initialized to $R = R_\varepsilon = T = T_\varepsilon = 0$:
\begin{itemize}
  \item If $t' < \ell$, we proceed as normal to draw $2^t$ samples and
    transition to a new state while accumulating rewards in variable $T$ as usual,
    and incrementing the number of rounds $R$ by one.
  \item If $t' = \ell$, then we enter a subgame where we play the
	general RANDAO game until we hit a state such that $(t'',0)$ for some
	$t''$. Add this reward to $T_\varepsilon$, increment $T$ by one, and
    add the additional number of rounds spent in the subgame to $R_\varepsilon$.
\end{itemize}
\end{definition}

We call this the \emph{subgame formulation} of the general game.
Clearly, this is a reformulation of the general RANDAO manipulation game where
the difference in reward from the original formulation is made more explicit.
Moreover, excluding the subgame, the actions of the adversary when $t' < \ell$
describes a policy for the RANDAO manipulation game $G$
which we use in the proof of the following theorem.

One caveat is that we now ignore the number of rounds it would take us in the subgame,
but this only increases the average reward.

\begin{theorem}
  Suppose a policy $\pi$ gets average reward $\Gamma$ in the general
  RANDAO manipulation game $G_r$.
  Then, there exists a policy $\pi^*$ in the RANDAO manipulation game $G$
  such that $\Gamma - \Gamma^* \leq X(t,\ell) (2\alpha)^{\ell} \ell$ where
  $\Gamma^*$ is the reward of $\pi^*$ in $G$ and $X(t,\ell)$
  is defined in Lemma~\ref{lem:reset-bound}.
\end{theorem}
\begin{proof}
  We consider a policy $\pi$ for the subgame variant and construct a policy $\pi^*$
  for the interleaved instance of the RANDAO manipulation game $G$.
  Suppose we are at state $(t,t')$. If $t' < \ell$, then the action of the policy is also
  available in the interleaved instance of the original game so we can also make the same
  choice for $\pi^*$.

  If $t' = \ell$, then after playing the subgame, $\pi$ transitions from $(t,\ell) \to^* (t'',0)$
  with reward $\varepsilon$ in several steps. Now, $\pi^*$ instead transitions from
  $(t,\ell) \to (\ell,t^*)$ with some reward $\lambda$. Note that for this transition $\pi^*$
  can be defined to pick any one of the options. Since we are establishing a bound on the difference,
  let $\lambda = 0$. In the next step, all options that are available
  to $\pi$ at $(t'',0)$ is also available to $\pi^*$ with identical distributions.
  Hence, we define $\pi^*$ such that $\pi^*$ chooses the action that $\pi$ would have
  chosen from the first $2^{t''}$ samples. Now, consider the value of each variable at time $m$.
  For notational convenience, we write $V$ instead of $V(m)$ for variable $V$.
  Using the existence of the limits,
  \begin{alignat*}{2}
      \Gamma - \Gamma^*
      &\leq \lim_{m \to \infty} \left( \frac{T + T_\varepsilon}{R + R_\varepsilon} \right)
         - \lim_{m \to \infty} \left( \frac{T}{R} \right)
      &&= \lim_{m \to \infty} \left( \frac{T + T_\varepsilon}{R + R_\varepsilon}
         - \frac{T}{R} \right) \\
      &= \lim_{m \to \infty} \left( \frac{T_\varepsilon R - T R_\varepsilon}{R(R + R_\varepsilon)} \right)
      &&\leq \lim_{m \to \infty} \left( \frac{T_\varepsilon}{R + R_\varepsilon} \right) \\
      &\leq \lim_{m \to \infty} \left( \frac{T_\varepsilon}{R} \right)
      &&\leq \ell \lim_{m \to \infty} \left( \frac{R_\varepsilon}{R} \right)
  \end{alignat*}
  Let $n_\ell(m)$ be the number of rounds until time $m$ that $\pi^*$ hits a state with $(\cdot, \ell)$.
  Hence, using Lemma~\ref{lem:reset-bound}, we have that
  $\lim_{m \to \infty} \left( R_\varepsilon/R \right) \leq X(t,\ell) \lim_{m \to \infty} (n_\ell / R)$.

  Observe that once we reach $(t,\ell)$ for some $t$, we transition back effectively into $(t',0)$.
  Since at any $(t,t')$ state with $t' < \ell$, we have $\leq (2\alpha)^\ell$ probability
  of transitioning to a state of the form $(t',\ell)$ by first applying a union bound
  to $2^{t'}$ events with success probability $\alpha^\ell$ and using $2^{t'} \leq 2^\ell$,
  we can simply consider the following two state process in order to find an upper bound on
  $\lim_{m \to \infty} (n_\ell / R)$. This is because we are overestimating the transition
  probability to a state of the form $(\cdot,\ell)$, and underestimating the transition
  probability back to a state of the form $(\cdot,t)$ for $t < \ell$.
  \begin{figure}[htp]
    \centering
    \begin{tikzpicture}[shorten >=1pt,node distance=2.5cm,auto]
    \tikzstyle{every state}=[fill={rgb:black,1;white,10}]

    \node[state]   (x)                      {$(\cdot,\ell)$};
    \node[state] (y) [right of=x]    {$(\cdot,t)$};

    \path[->]
    (x) edge [bend left]    node {$1-(2\alpha)^{\ell}$}   (y)
    (x) edge [loop left]    node {$(2\alpha)^{\ell}$}     ()
    (y) edge [bend left]	node {$(2\alpha)^{\ell}$}     (x)
    (y) edge [loop right]   node {$1-(2\alpha)^{\ell}$}   ();
  \end{tikzpicture}
  \end{figure}
  Solving for the stationary distribution,
  we see that the fraction of the state $(\cdot,\ell)$ is $(2\alpha)^\ell$.
  Therefore we get $\Gamma - \Gamma^* \leq \ell X(t,\ell) (2\alpha)^{\ell}$.
\end{proof}

Plotting this bound for $\alpha$ such that $(2\alpha)^{\ell} + (1-\alpha^\ell) (4\alpha)^\ell < 1$,
we get the following figure:
\begin{figure}[htp]
    \centering
    \includegraphics[width=0.75\textwidth]{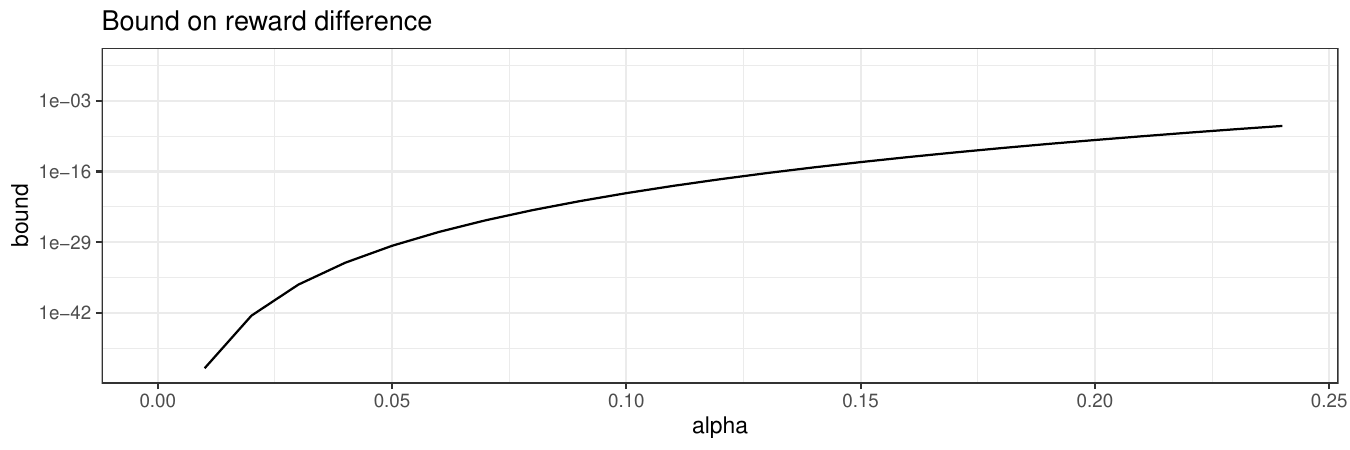}
    \caption{Bound on the error induced by simplifying the game}
\end{figure}

Interpreting the error in the estimated average reward, we can observe that it stays reasonably small.
For example, roughly speaking, for $\alpha = 0.05$, we can expect the difference to be smaller
than $\leq 1/10^{30}$, for $\alpha = 0.1$, we can expect the bound to be $\leq 1/10^{20}$,
for $\alpha = 0.2$, $\leq 1/10^{11}$ epochs, for $\alpha = 0.24$, $\leq 1/10^{8}$ epochs, etc.
\end{toappendix}

\end{document}